\documentclass{article}

\usepackage[affil-it]{authblk}
\usepackage[dvipsnames]{xcolor}
\usepackage{amsfonts}
\usepackage{amsmath,amsthm,amssymb,dsfont}

\usepackage{enumerate}
\usepackage{graphicx}	
\usepackage{subcaption}
\usepackage[margin=3cm]{geometry}
\usepackage{url}
\usepackage{todonotes}
\usepackage{bbm}

\usepackage{tikz}
\usepackage{circuitikz}

\usetikzlibrary{arrows.meta} 
\tikzset{ meter/.append style={ draw, inner sep=10, rectangle, font=\vphantom{A}, minimum width=30, line width=.5, path picture={ \draw[black] ([shift={(.1,.3)}]path picture bounding box.south west) to[bend left=50] ([shift={(-.1,.3)}]path picture bounding box.south east); \draw[black,-{Latex[scale=0.6]}] ([shift={(0,.1)}]path picture bounding box.south) -- ([shift={(.3,-.1)}]path picture bounding box.north); } } }

\usepackage{pifont}
\usepackage{multirow}
\usepackage{makecell}

\usepackage{epsfig}
\usetikzlibrary{shapes.symbols,patterns} 
\usepackage{pgfplots}
\pgfplotsset{compat=1.10}
\usepgfplotslibrary{fillbetween}
\usetikzlibrary {decorations.pathmorphing, decorations.pathreplacing, decorations.shapes}

\definecolor{linkblue}{HTML}{001487}
\usepackage{hyperref}
\hypersetup{colorlinks=true,citecolor=linkblue,linkcolor=linkblue,filecolor=linkblue,urlcolor=linkblue,breaklinks=true}

\usepackage{nicefrac}
\usepackage{mathtools}

\usepackage{thmtools} 
\hypersetup{hypertexnames=false}

\usepackage{algorithm}
\usepackage{algorithmic}

\usepackage{mdframed}
\usepackage{aligned-overset}
\usepackage{circuitikz}

\theoremstyle{plain}
\newtheorem{theorem}{Theorem}[section]
\newtheorem{lemma}[theorem]{Lemma}

\newtheorem{corollary}[theorem]{Corollary}

\theoremstyle{definition}
\newtheorem{definition}[theorem]{Definition}

\newtheorem{non-example}[theorem]{Non-example}

\DeclareRobustCommand{\abbrevcrefs}{%
\Crefname{theorem}{Thm.}{Thms.}%
\Crefname{corollary}{Cor.}{Cors.}%
\Crefname{lemma}{Lem.}{Lems.}%
\Crefname{remark}{Rmk.}{Rmks.}%
\Crefname{proposition}{Prop.}{Props.}%
\Crefname{equation}{Eq.}{Eqs.}%
\Crefname{example}{Ex.}{Exs.}%
}

\DeclareMathOperator*{\argmin}{argmin}

\DeclareRobustCommand{\Cshref}[1]{{\abbrevcrefs\Cref{#1}}}

\newcommand*{\ee}{\mathrm{e}}

\newcommand*{\cA}{\mathcal{A}}
\newcommand*{\cB}{\mathcal{B}}

\newcommand*{\cE}{\mathcal{E}}
\newcommand*{\cF}{\mathcal{F}}
\newcommand*{\cG}{\mathcal{G}}
\newcommand*{\cH}{\mathcal{H}}

\newcommand*{\cI}{\mathcal{I}}

\newcommand*{\cN}{\mathcal{N}}
\newcommand*{\cM}{\mathcal{M}}
\newcommand*{\cP}{\mathcal{P}}
\newcommand*{\cQ}{\mathrm{Q}}

\newcommand*{\cS}{\mathcal{S}}
\newcommand*{\cT}{\mathcal{T}}
\newcommand*{\cU}{\mathcal{U}}
\newcommand*{\cV}{\mathcal{V}}

\newcommand*{\N}{\mathbb{N}}

\newcommand*{\R}{\mathbb{R}}

\newcommand*{\St}{\mathrm{S}}

\newcommand*{\reg}{\infty}

\newcommand*{\eps}{\varepsilon}
\newcommand*{\diag}{\mathrm{diag}}

\newcommand*{\rank}{\mathrm{rank}}

\newcommand*{\id}{\mathds{1}}

\newcommand*{\spec}{\mathrm{spec}}
\newcommand*{\supp}{\mathrm{supp}}
\newcommand*{\range}{\mathrm{ran}}
\newcommand*{\tr}{\mathrm{tr}}
\newcommand*{\ket}[1]{| #1 \rangle}
\newcommand*{\bra}[1]{\langle #1 |}

\newcommand{\proj}[1]{|#1\rangle\!\langle #1|}

\newcommand*{\D}{\mathrm{D}}

\newcommand*{\CPTP}{\mathrm{CPTP}}
\newcommand*{\CPTN}{\mathrm{CPTN}}
\newcommand*{\CP}{\mathrm{CP}}

\newcommand*{\ci}{\mathrm{i}} 
\newcommand*{\di}{\mathrm{d}} 

\newcommand{\norm}[1]{\left\lVert#1\right\rVert}

\usepackage{float}
\usepackage[nameinlink,capitalize,noabbrev]{cleveref}

 \allowdisplaybreaks

\definecolor{protocolblue}{HTML}{1F5A94}
\definecolor{protocolorange}{HTML}{C66A2B}
\definecolor{protocolgreen}{HTML}{2F7D68}
\definecolor{protocolgray}{HTML}{68717A}

\title{The regularized channel R\'enyi divergence is continuous} 

    \author{\normalsize Lukas Schmitt$^{1,2}$ and David Sutter$^{2}$}
     \affil{\small $^{1}$Institute for Theoretical Physics, ETH Zurich\\
     $^{2}$IBM Research Europe -- Zurich}
\date{}

\begin{document}

\maketitle

\begin{abstract}
We show that the regularized channel R\'enyi divergence is continuous at $\alpha = 1$. As consequences, we establish an exponentially strong converse for channel discrimination, a subchannel-smoothed asymptotic equipartition property, a relative entropy accumulation theorem, and a new proof of the generalized quantum Stein's lemma. Our result further shows that the regularized channel relative entropy is computable. The continuity proof relies on a variational formula that expresses the regularized channel divergence as a single-letter optimization over variables of unrestricted finite dimension. We then establish continuity using a novel spectral confinement technique. 
\end{abstract}


\section{Introduction}
Quantum relative entropy describes the asymptotic rate at which two states can be distinguished~\cite{PH91,ogawa00}. 
For quantum channels, a similar task can be formulated, which yields the channel relative entropy~\cite{FFRS19,Wang2019,beta18}. However, unlike in state discrimination, input states can be entangled with a reference system and correlated over many channel uses. In technical terms, the channel relative entropy is not additive under the tensor product and therefore needs a regularization~\cite[Section~3.1]{FFRS19}. 

R\'enyi divergences form a prominent family of generalizations of relative entropy, parameterized by the real number $\alpha \in [\frac{1}{2},\infty)$. In the limit $\alpha \to 1$ the R\'enyi divergence converges to the relative entropy~\cite{petz86,MLDSFT13,WWY14}. Due to the regularization, it is unclear if the regularized channel R\'enyi divergence converges to the regularized channel relative entropy in the limit $\alpha \to 1$. 
In this work, we show that this is indeed the case.

\begin{theorem}\label{thm:main}
Let $\cE$ be a completely positive and trace-preserving map and $\cF$ be a completely positive map. Then, 
\begin{align}\label{eq:main}
\lim_{\alpha\downarrow 1} D_\alpha^\infty(\cE\|\cF)
=D^\infty(\cE\|\cF)\, .
\end{align}
\end{theorem}
\begin{proof}[Proof overview]
We prove the assertion for completely positive and trace-preserving (CPTP) maps $\cE$, $\cF$ and then afterwards show that this also implies the statement for the case where $\cF$ is only CP (see~\cref{sec_CP_sufficient}).
A key ingredient of the proof is a novel variational expression \smash{$C_{(\alpha -1)/\alpha}(\cE,\cF)$} with $\alpha \in [1,\infty)$\footnote{see~\cref{eq_def_C}} for the regularized channel R\'enyi divergence $D_\alpha^\infty(\cE\|\cF)$\footnote{see~\cref{eq:regularized}}, that is a single-letter optimization problem, although over variables of unbounded reference dimension. A spectral confinement argument then yields error bounds independent of the reference dimension.
With this we show that
\begin{align} \label{eq_overview_proof}
\lim_{\alpha\downarrow 1}D_\alpha^\infty(\cE\|\cF) 
\overset{\textnormal{\Cshref{lem_var_formula_new}}}{=} \lim_{s \downarrow 0} C_s(\cE,\cF) 
\overset{\textnormal{\Cshref{lem_cont_Cs}}}{=} C_0(\cE,\cF) 
\overset{\textnormal{\Cshref{lem_var_formula_new}}}{=}  D^\infty(\cE\|\cF) \, .
\end{align}
The details are given in~\cref{sec_proof_main}.
\end{proof}
The limit 
\begin{align}\label{eq:simple_direction}
\lim_{\alpha\uparrow 1} D_\alpha^\infty(\cE\|\cF) = D^\infty(\cE\|\cF)
\end{align}
can be obtained from standard techniques. This is explained in~\cref{app_justification_simple_limit}.
Furthermore, note that by~\cite[Corollary 3.6]{Jencova2018Renyi}, continuity for the (sandwiched) R\'enyi divergence also implies continuity for the Petz R\'enyi divergence~\cite{petz86}.

\paragraph{Applications} 
\Cref{thm:main} has a few consequences, which are discussed in more detail in~\cref{sec_applications}. It implies
\begin{enumerate}[(a)]
\item that $D^\infty(\cE\|\cF)$ is computable (see~\cref{sec_computability});
\item that $D^\infty(\cE\|\cF)$ is the Stein exponent for channel discrimination (see~\cref{sec_channel_disc});
\item a channel asymptotic equipartition property (AEP), where the smoothing is over subchannels and with respect to (generalized) diamond norm (see~\cref{sec_channel_AEP});
\item a relative entropy accumulation theorem (see~\cref{sec_REAT});
\item a proof for the generalized quantum Stein's lemma, with a more general null hypothesis (see~\cref{sec_GSL}).
\end{enumerate}
\section{Notation}
All logarithms are natural logarithms and all Hilbert spaces are finite-dimensional in this paper. Powers are taken on the support and $\mathbf{1}_{[a,b]}(X)$ denotes the projection onto the eigenspace corresponding to the eigenvalues in $[a,b]$. Furthermore, we denote the range of an operator $X$ by $\range(X)$ and its support by $\supp(X)$. The notation $\rho \ll \sigma$ means $\supp(\rho) \subseteq \supp(\sigma)$.
Let $\St(\cH)$ denote the set of density matrices on $\cH$ and let $\CPTP(A,B)$ be the set of completely positive and trace-preserving maps from $A$ to $B$. The set of completely positive and trace non-increasing maps is denoted by $\CPTN(A,B)$. 
For $\cE,\cF \in \CPTP(A,B)$, let  $\cE_R:=\cI_R\otimes\cE$ and $\cF_R:=\cI_R\otimes\cF$, where $\cI_R$ denotes the identity map on $R$. For $\cE \in \CPTP(A,B)$ its adjoint map is denoted by $\cE^*$.

For $\alpha \in (1,\infty)$ the \emph{(sandwiched) R\'enyi divergence}~\cite{MLDSFT13,WWY14} is given by
\begin{align} \label{eq_sandwhiched}
D_{\alpha}(\rho \| \sigma) :=\frac{1}{\alpha -1 } \log \tr\big[(\sigma^{\frac{1-\alpha}{2\alpha}} \rho \, \sigma^{\frac{1-\alpha}{2\alpha}})^{\alpha}\big] \, ,
\end{align}
if $\rho \ll \sigma$, and $+\infty$ otherwise.
In the limits $\alpha \to 1$ and $\alpha \to \infty$ the (sandwiched) R\'enyi divergence converges to the relative entropy $D(\rho \| \sigma):=\tr[\rho(\log \rho - \log \sigma)]$ if $\rho \ll \sigma$, and $+\infty$ otherwise, and the max-relative entropy~\cite{renner_phd,datta09}
\begin{align}
D_{\max}(\rho \| \sigma) :=\inf\{\lambda \in \R: \rho \leq \ee^\lambda \sigma \} \, ,
\end{align}
respectively.
The measured R\'enyi divergence is defined as
\begin{align}
D_{\alpha,\mathbb{M}}(\rho \| \sigma):=\sup_{\cM} D_{\alpha}\big(\cM(\rho) \| \cM(\sigma) \big) \, ,
\end{align}
where the supremum is over all POVM channels $\cM$~\cite{donald86,PH91,FBT15}.

The stabilized channel divergence is
\begin{align}
D_{\alpha}(\cE\|\cF) &=\sup_{\ket{\psi}_{RA}} D_{\alpha}\big(\cE_R(\psi)\|\cF_R(\psi)\big)\, , \label{eq:channel}
\end{align}
where purification and data-processing allow us to restrict to pure inputs and a reference of dimension at most the input dimension~\cite{LKDW18}. 
Its regularization is given by
\begin{equation} 
D_{\alpha}^\infty(\cE\|\cF) 
= \lim_{n\to\infty}\frac{1}{n}D_{\alpha}(\cE^{\otimes n}\|\cF^{\otimes n})
=\sup_{n\geq 1}\frac{1}{n}D_{\alpha}(\cE^{\otimes n}\|\cF^{\otimes n})\, ,\label{eq:regularized}
\end{equation}
where the limit exists and can be expressed as a supremum by Fekete's lemma~\cite{fekete23}.
Let $J_{\cE}$ and $J_{\cF}$ be the Choi matrices of $\cE$ and $\cF$, respectively.\footnote{The Choi matrix is defined as $J_{\cE}:=\cE_R(\Phi_{AR})$, where $\Phi_{AR}:=\sum_{ar}\ket{a}\bra{r}_A \otimes \ket{a}\bra{r}_R$} Note that $D_{\max}(\cE \| \cF) = D_{\max}(J_{\cE} \| J_{\cF})$, which is finite if and only if $J_{\cE} \ll J_{\cF}$.
In the case $J_{\cE} \not \ll J_{\cF}$ both sides of~\cref{eq:main} are unbounded. Thus, we assume without loss of generality that $J_{\cE} \ll J_{\cF}$.

The channel R\'enyi divergence satisfies a chain rule~\cite{FFRS19,FF21} which will be crucial for this work. For any $\alpha \in [1,\infty)$, $\rho\in \St(A)$, $\sigma \geq 0$, and $\cE,\cF \in \CPTP(A,B)$ we have
\begin{align} \label{eq_chain_rule}
D_{\alpha}\big(\cE(\rho)\| \cF(\sigma)\big) \leq D_{\alpha}(\rho \| \sigma) + D^{\infty}_\alpha(\cE \| \cF) \, .
\end{align}

For $\eps\in(0,1)$ we define the $\eps$-ball around $\rho$ by $\cB_{\eps}(\rho):=\{\rho' \in \St(\cH): P(\rho,\rho') \leq \eps \}$, where \smash{$P(\rho,\sigma):=\sqrt{1-F(\rho,\sigma)}$} is the purified distance~\cite{marco_book} and $F(\rho, \sigma):=\norm{\sqrt{\rho} \sqrt{\sigma}}^2_1$ is the fidelity. 
The smooth max-relative entropy is defined as
\begin{align}
D_{\max}^\eps(\rho  \| \sigma ) :=  \min_{\rho' \in \cB_{\eps}(\rho)} D_{\max}(\rho' \| \sigma) \, . 
\end{align} 

\section{Proof of~\cref{thm:main}} \label{sec_proof_main}
We prove the assertion of~\cref{thm:main} first for $\cE,\cF \in \CPTP$ and lift to $\cF \in \CP$ in the end. 
For a positive definite matrix $T>0$ and $s \in (0,1)$ we define
\begin{align} \label{eq_def_q}
q_s(T,\rho)
:=\frac{1}{s}\log \frac{\tr[\cE_R(\rho)T^s]}{\tr[\rho \cF_R^*(T)^s]} \qquad \textnormal{and} \qquad
q_0(T,\rho):=\tr[\cE_R(\rho)\log T]-\tr[\rho\log \cF_R^*(T)] \, ,
\end{align}
as well as
\begin{align} \label{eq_def_C}
C_s(\cE,\cF):=\sup_{T>0,\rho \in \St(AR)} q_s(T,\rho) \qquad \textnormal{and} \qquad 
C_0(\cE,\cF):=\sup_{T>0,\rho \in \St(AR)} q_0(T,\rho)\, .
\end{align}
Note that the matrix $\cF_R^*(T)$ is strictly positive because $\cF_R^*$ is unital and $\cF$ is CPTP.
Every reference system $R$ in these suprema is finite, but its dimension is unrestricted.
We next study some mathematical properties of these quantities.
\begin{lemma}\label{lem:properties}
Let $\cE,\cF \in \CPTP(A,B)$ and $s \in [0,1)$. Then $0\leq C_s(\cE,\cF)\leq D_{\max}(\cE \| \cF)$ and $C_r(\cE,\cF)\leq C_s(\cE,\cF)$ for $0\leq r<s<1$. Furthermore, pure input states suffice in~\cref{eq_def_C} and we have
\begin{equation}
    C_s(\cE,\cF) = \frac{1}{s} \!\sup_{T>0}D_{\max}\big(\cE_R^*(T^s)\|(F_R^*T)^s\big) \!\quad \!\textnormal{and}\! \quad \!C_0(\cE,\cF) = \sup_{T>0} \lambda_{\max}\big( \cE_R^*(\log T) - \log \cF^*_R(T)\big) \, .
\end{equation}
In addition for $\cE_i,\cF_i \in \CPTP(A_i,B_i)$ with $i \in \{ 1,2\}$,
\begin{align}\label{eq:tensorization}
C_s(\cE_1\otimes\cE_2,\cF_1\otimes\cF_2) = C_s(\cE_1,\cF_1)+C_s(\cE_2,\cF_2) \, .
\end{align}
\end{lemma}
The proof of~\cref{lem:properties} is given in~\cref{app_proof_properties}.

\begin{lemma}[Variational expression for regularized channel R\'enyi divergence]\label{lem_var_formula_new}
Let $\cE,\cF \in \CPTP(A,B)$ and $s \in (0,1)$. Then
    \begin{align}
        C_s(\cE,\cF) =  D_{\frac{1}{1-s}}^\infty(\cE\|\cF)  \qquad \textnormal{and} \qquad C_0(\cE,\cF) = D^\infty(\cE\|\cF) \, .
    \end{align}
\end{lemma}
\Cref{lem_var_formula_new} is interesting as it allows us to express a regularized channel divergence with an expression that contains no regularization but involves an optimization over reference systems of arbitrary finite dimension. The proof is given in~\cref{sec_proof_var_formula}.

\begin{lemma}[Continuity of variational expression]\label{lem_cont_Cs}
Let $\cE,\cF \in \CPTP(A,B)$. Then 
    \begin{align}
         \lim_{s \downarrow 0} C_s(\cE,\cF) = C_0 (\cE,\cF) \, .
    \end{align}
\end{lemma}
The proof of~\cref{lem_cont_Cs} is given in~\cref{sec_proof_confinement}.
It develops a spectral confinement method showing that almost optimal values for $C_s(\cE,\cF)$ and $C_0(\cE,\cF)$ can be attained by inputs supported on a controlled interval. This restriction gives an error bound that is uniform in the reference dimension, without requiring a lower bound on the nonzero eigenvalues of the input state.


\subsection{Proof of~\cref{lem:properties}}
\label{app_proof_properties}
Without loss of generality we may assume that the channels $\cE,\cF$ are such that $D_{\max}(\cE \| \cF) <\infty$, as otherwise the quantities are unbounded and the statements from~\cref{lem:properties} hold trivially.
For $T>0$, let $0<A = \cF_R^*(T)$, $N_s(T)=A^{-s/2}\cE_R^*(T^s)A^{-s/2}$ and $\omega=\frac{A^{s/2}\rho A^{s/2}}{\tr[\rho A^s]}$.
Cyclicity of the trace gives
\begin{equation}\label{eq:tilt}
\tr[\omega N_s(T)] = \frac{\tr[\rho\cE_R^*(T^s)]}{\tr[\rho A^s]} = \ee^{s q_s(T,\rho)}.
\end{equation}
Maximizing the left-hand side of~\cref{eq:tilt} over all density matrices yields a top eigenvector of $N_s(T)$, i.e.~we find
\begin{equation}
\frac{1}{s} \log \lambda_{\max}(N_s(T)) = \sup_{\rho \in \St(R  A)} q_s(T,\rho) \leq C_s(\cE,\cF) \, .
\end{equation}
This implies that $N_s(T) \leq \ee^{s C_s(\cE,\cF)} \id$ and hence $C_s(\cE,\cF)$ is the optimal constant such that
\begin{align}
\cE_R^*(T^s)&\leq \ee^{s C_s(\cE,\cF)}(\cF_R^*(T))^s \, \label{eq_s_big}
\end{align}
for all tests $T>0$.
For $s=0$, recall that $q_0(T,\rho)=\tr[\cE_R(\rho) \log T] - \tr[\rho \log A]$ and hence
\begin{align}
C_0(\cE,\cF) 
\geq \sup_{\rho \in \St(R A)} q_0(T,\rho)
=\lambda_{\max}(\cE^*_R(\log T) - \log \cF_R^*(T)) \, ,
\end{align}
which implies that $C_0(\cE,\cF)$ is the optimal constant satisfying
\begin{align}
    \cE_R^*(\log T)\leq\log(\cF_R^*(T))+C_0(\cE,\cF)\id \, .\label{eq_s_zero}
\end{align}
Unital operator Jensen~\cite{hansen03} and operator monotonicity of $x \mapsto x^s$~\cite[Table~2.2]{Sutter_book} give\footnote{Note that $J_{\cE} \leq J_{\cF}$ implies $J_{\cE^*} \leq J_{\cF^*}$ and hence $D_{\max}(\cE \| \cF) = D_{\max}(\cE^* \| \cF^*)$.}
\begin{align}
\cE_R^*(T^s) \leq(\cE_R^*(T))^s \leq \ee^{sD_{\max}(\cE\|\cF)}(\cF_R^*T)^s
\end{align}
and
\begin{align}
    \cE_R^*(\log T)
    \leq \log(\cE_R^*(T))
    \leq\log( \exp(D_{\max}(\cE\|\cF)) \cF_R^*(T))
    =\log(\cF_R^*(T))+ D_{\max}(\cE\|\cF) \id \, .
\end{align}
Taking $T=\id$ gives the lower bound zero. This together with~\cref{eq_s_big,eq_s_zero} proves that $0 \leq C_s(\cE,\cF) \leq D_{\max}(\cE\|\cF)$.
For $0<r<s<1$, a further application of Jensen leads to
\begin{equation}\label{eq:monotonicity}
    \cE_R^*(T^r) \leq \bigl(\cE_R^*(T^s)\bigr)^{r/s} 
    \overset{\textnormal{\Cshref{eq_s_big}}}{\leq}\bigl(\ee^{sC_s(\cE,\cF)}(\cF_R^*(T))^s\bigr)^{r/s}
    =\ee^{r C_s(\cE,\cF)}(\cF_R^*(T))^r \, ,
\end{equation}
which implies $C_r\leq C_s$ by the characterization as the smallest constant.
For the case of $r=0$, applying Jensen to the logarithm and using~\eqref{eq_s_big} gives
\begin{equation}
    \cE_R^*(\log T) = \frac{1}{s}\cE_R^*(\log T^s) \leq \frac{1}{s}\log\cE_R^*(T^s) \leq \log A+C_s(\cE,\cF) \id \, ,
\end{equation}
which by comparing the expression with \eqref{eq_s_zero} then gives $C_0(\cE,\cF) \leq C_s(\cE,\cF)$.

To prove additivity under tensor products, we treat the first output system as a reference when applying the inequality to the second channel. Applying \cref{eq_s_big} twice gives
\begin{align}
    (\cE_1^*\!\otimes\!\cE_2^*)(T^s)
    \leq \ee^{s C_s(\cE_2,\cF_2)}(\cE_1^*\!\otimes\!\id) \Big(\bigl((\id\!\otimes\!\cF_2^*)(T)\bigr)^s\Big)
    \leq \ee^{s(C_s(\cE_1,\cF_1)\!+\!C_s(\cE_2,\cF_2))}\bigl((\cF_1^*\!\otimes\!\cF_2^*)(T)\bigr)^s.
\end{align}
The same calculation at $s=0$, using~\cref{eq_s_zero},  yields
\begin{align}
    (\cE_1^*\otimes\cE_2^*)(\log T)
    &\leq(\cE_1^*\otimes\id) \Big(\log\bigl((\id\otimes\cF_2^*)(T)\bigr)\Big)+C_0(\cE_2,\cF_2)\id \\
    &\leq \log\bigl((\cF_1^*\otimes\cF_2^*)(T)\bigr) +(C_0(\cE_1,\cF_1)+C_0(\cE_2,\cF_2)) \id\, .
\end{align}
All these inequalities remain valid with an additional reference. To show the other direction, we take product tests and product input states. \qed

\subsection{Proof of~\cref{lem_var_formula_new}} \label{sec_proof_var_formula}
To prove~\cref{lem_var_formula_new} we need a few preparatory results. We start by recalling known variational expressions for the measured R\'enyi divergence.
\begin{lemma}[\cite{FBT15}]\label{lem:measured-formulas}
Let $\rho \in \St(\cH)$, $\sigma\geq0 $, and $\alpha>1$. Then
\begin{align}
    D_{\alpha,\mathbb{M}}(\rho\|\sigma)&=\sup_{X>0}\left\{\frac{\alpha}{\alpha -1}\log\tr[\rho X^\frac{\alpha -1}{\alpha}]-\log\tr[\sigma X]\right\},\label{eq:var1}\\
    \ee^{(\alpha-1) D_{\alpha,\mathbb{M}}(\rho\|\sigma)} &=\sup_{X>0}\left\{\alpha\tr[\rho X^\frac{\alpha -1}{\alpha}]-(\alpha-1)\tr[\sigma X]\right\},\label{eq:var2}\\
    D_{\mathbb{M}}(\rho\|\sigma) &=\sup_{X>0}\left\{\tr[\rho\log X]+1-\tr[\sigma X ]\right\}.\label{eq:var3}
\end{align}
\end{lemma}
We will often use the reparameterization $\alpha = \frac{1}{1-s}$ to switch between $\alpha \in [1,\infty)$ and $s \in [0,1)$.
For $X > 0$, $\rho \in \St(A R)$, and $\sigma \geq 0$, define the function
\begin{align}\label{eq:Vs}
V_s(X,\rho,\sigma):=\frac{1}{s} \log \tr[\rho X^s]-\log\tr[\sigma X] \qquad \textnormal{for } s \in (0,1)
\end{align}
and
\begin{align}
    V_0(X,\rho,\sigma)=\tr[\rho\log X]+1-\tr[\sigma X] \, .
\end{align}
By definition of the adjoint map we have $\tr[\cF_R(\sigma)T]=\tr[\sigma \cF_R^*(T)]$ and thus for every $s \in [0,1)$ 
\begin{align}\label{eq:test-cancellation}
    V_s(T,\cE_R(\rho),\cF_R(\sigma))-V_s(\cF_R^*(T),\rho,\sigma)=q_s(T,\rho)\, .
\end{align}

\begin{lemma}\label{lem:derivative}
Let $A>0$, $\rho \in \St(\cH)$, and $s \in [0,1)$. Define
\begin{equation}
    \sigma_s=f_{s,A}(\rho):= \begin{cases}
    \frac{1}{s} \frac{\di }{\di t} \big(A+t \rho\big)^s \big|_{t=0},&0<s<1,\\
    \frac{\di}{\di t} \log(A + t \rho) \big|_{t=0},&s=0\,. \end{cases}
\end{equation}
Then $\sigma_s\geq 0 $, $\rho\ll\sigma_s$, $\tr[\sigma_s]=\tr[\rho A^{s-1}]$ and $\tr[\sigma_s A]=\tr[\rho A^s]$.
The test $A$ attains the suprema in~\cref{eq:var1,eq:var2} for $(\rho,\sigma_s)$. At $s=0$ it attains the supremum in~\cref{eq:var3} for $(\rho,\sigma_0)$.
In particular, using the notation from~\cref{lem:measured-formulas}
\begin{equation}
    V_s(A,\rho,\sigma_s) = D_{\alpha,\mathbb{M}}(\rho\|\sigma_s)
\end{equation}
\end{lemma}
\begin{proof}
For every $s \in[0,1)$ recall the following integral representation (see for instance~\cite{HiaiPetz2014})
\begin{equation}\label{eq:f_s-integral}
    f_{s,A}(X)=\int_0^\infty h_s(t)(A+t \id )^{-1}X(A+t \id)^{-1}\,\di t
    \quad \textnormal{with} \quad
    h_s(t) = \begin{cases} \dfrac{\sin(\pi s)}{\pi s}t^s,&s>0\\
    1,&s=0 \, . \end{cases}
\end{equation}
It shows that $\sigma_s\geq 0$ and $\tr[Xf_{s,A}(Y)]=\tr[f_{s,A}(X)Y]$ for any $X,Y$.
Since $\id$ and $A$ commute, we have $f_{s,A}(\id)=A^{s-1}$ and $f_{s,A}(A)=A^s$.
Hermiticity then gives
\begin{align}
    \tr[\sigma_s]=\tr[\rho f_{s,A}(\id)]=\tr[\rho A^{s-1}] \qquad \textnormal{and} \qquad 
    \tr[\sigma_s A]=\tr[\rho f_{s,A}(A)]=\tr[\rho A^s] \, .
\end{align}\label{eq:eq_A}
We next prove that $A$ maximizes~\cref{eq:var2} and then show that it also attains the supremum in~\cref{eq:var1}. For $s>0$, define $F(U)=\alpha\tr[\rho U^s]-(\alpha-1)\tr[\sigma_s U]$.
Its derivative at $A$ in Hermitian direction $X$ then is
\begin{align}
    \frac{\di}{\di t}F(A+tX) \big|_{t=0} =\alpha s \,  \tr[f_{s,A}(\rho)X]-(\alpha-1)\tr[\sigma_s X]=0.
\end{align}
Here we used again Hermiticity and that $\alpha s=\alpha-1$. Since $F$ is concave, $A$ is indeed a global maximizer and optimizes~\cref{eq:var2}. Since we have $F(A) = \tr[\rho A^s] = \tr[\sigma_sA] = \ee^{(\alpha-1)D_{\alpha,\mathbb{M}}(\rho\|\sigma_s)}$,
we get
\begin{equation}
    V_s(A,\rho,\sigma_s) = \frac{1}{s} \log \tr[\rho A^s] - \log \tr[\sigma_s A] = 
    \left(\frac{1}{s}-1\right)\log  F(A) = \frac{1}{\alpha-1}\log F(A) = D_{\alpha,\mathbb{M}}(\rho\|\sigma_s).
\end{equation}
For $s=0$, the objective in~\cref{eq:var3} is concave.
Its derivative at $A$ in any Hermitian direction $X$ is 
\begin{equation}
     \frac{\di}{\di t} V_0(A+tX,\rho,\sigma_0) \Bigr|_{t=0} = \tr[f_{0,A}(\rho)X]-\tr[\sigma_0 X] = 0 \, .
\end{equation}
Hence $A$ is a global maximizer, which gives
\begin{equation}
    V_0(A,\rho,\sigma_0) = D_{\mathbb{M}}(\rho\|\sigma_0)\,.
\end{equation}
Since the optimal $A$ is finite, we get $D_{\alpha,\mathbb{M}}(\rho\|\sigma_s) < \infty$ which also shows $\rho \ll \sigma_s$.
\end{proof}

\begin{lemma}\label{lem:measured-amortization}
Let $s \in [0,1)$ and $\cE,\cF \in \CPTP(A,B)$. Then,
\begin{align}
    C_s(\cE,\cF)=\sup_{\rho,\sigma \in \St(A R):\rho\ll\sigma} \left\{D_{\frac{1}{1-s},\mathbb{M}}\big(\cE_R(\rho)\|\cF_R(\sigma)\big)-D_{\frac{1}{1-s},\mathbb{M}}(\rho\|\sigma)\right\}.
\end{align}
\end{lemma}
\begin{proof}
We start by noting that for any $T >0$ and $A=\cF_R^*(T)$
\begin{align}
    V_s(T,\cE_R(\rho),\cF_R(\sigma))-D_{\frac{1}{1-s},\mathbb{M}}(\rho\|\sigma)
    \overset{\textnormal{\Cshref{lem:measured-formulas}}}&{\leq} V_s(T,\cE_R(\rho),\cF_R(\sigma))-V_s(A,\rho,\sigma)\\
    \overset{\textnormal{\Cshref{eq:test-cancellation}}}&{=} q_s(T,\rho) \\
    \overset{\textnormal{\Cshref{eq_def_C}}}&{\leq} C_s(\cE,\cF) \, .
\end{align}
Optimizing over all $T>0$ and using~\cref{lem:measured-formulas} yields
\begin{align}
    C_s(\cE,\cF)\geq \sup_{\rho,\sigma \in \St(AR):\rho\ll\sigma} \left\{D_{\frac{1}{1-s},\mathbb{M}}\big(\cE_R(\rho)\|\cF_R(\sigma)\big)-D_{\frac{1}{1-s},\mathbb{M}}(\rho\|\sigma)\right\}.
\end{align}

To see the other direction, fix $T>0$, a density matrix $\rho$, and choose $A = \cF_R^*(T)$ and $\sigma_s=f_{s,A}(\rho)$. By \cref{lem:derivative}, $A$ is now an optimal input test and hence
\begin{align}
    D_{\frac{1}{1-s},\mathbb{M}}\big(\cE_R(\rho)\|\cF_R(\sigma_s))-D_{\frac{1}{1-s},\mathbb{M}}(\rho\|\sigma_s)
    \overset{\textnormal{\Cshref{lem:derivative}}}&{\geq} V_s(T,\cE_R(\rho),\cF_R(\sigma_s))-V_s(A,\rho,\sigma_s) \\
    \overset{\textnormal{\Cshref{eq:test-cancellation}}}&{=}q_s(T,\rho)\, .\label{eq:amort-lower}
\end{align}
Replacing $\sigma_s$ by $\sigma_s/\tr[\sigma_s]$ adds $\log \tr \sigma_s$ to both divergences and therefore leaves their difference unchanged.
Since this holds for all $T>0$ and all density matrices $\rho$ we find
\begin{align}
    C_s(\cE,\cF)\leq \sup_{\rho,\sigma \in \St(AR):\rho\ll\sigma} \left\{D_{\frac{1}{1-s},\mathbb{M}}\big(\cE_R(\rho)\|\cF_R(\sigma)\big)-D_{\frac{1}{1-s},\mathbb{M}}(\rho\|\sigma)\right\}\, .
\end{align}
\end{proof}

\begin{lemma}\label{lem:pinch}
        Let $\rho \in \St(\cH)$, $\sigma\geq0$ and $\alpha\geq 1$. If $\sigma$ has $v$ distinct positive eigenvalues, we have
    \begin{align} \label{eq_piching}
        D_\alpha(\rho \| \sigma) \leq  D_{\alpha,\mathbb{M}}(\rho \| \sigma) + \log v \, .
    \end{align}
\end{lemma}
\begin{proof}
Without loss of generality we can assume that $\rho \ll \sigma$ as otherwise~\cref{eq_piching} is trivial.
    Suppose the spectral decomposition of $\sigma$ is given by $\sigma = \sum_{i=1}^v q_i P_i$. Then we can define a pinching with respect to $\sigma$ by $\cP_\sigma(X) = \sum_{i=1}^v P_i X P_i$.\footnote{More information about the pinching map can be found in~\cite[Section~3.1]{Sutter_book}.}
    Now define the matrix 
    \begin{equation}
        X = \sigma^{\frac{1-\alpha}{2\alpha}} \rho \sigma^{\frac{1-\alpha}{2\alpha}} \, .
    \end{equation}
    Since $\rho \ll \sigma$, we can restrict ourselves to $\supp \, \sigma$ and get $X = \sum_i X^{1/2}  P_i X^{1/2} =: \sum_i Y_i$. By convexity of $X \mapsto \tr[X^\alpha]$~\cite[Proposition~2.10]{Sutter_book}, we have
    \begin{equation}
        \tr[X^\alpha] 
        = v^\alpha \tr\left[\left(\frac{1}{v}\sum_{i=1}^v Y_i\right)^\alpha\right] 
        \leq v^\alpha \frac{1}{v} \sum_{i=1}^v \tr[Y_i^\alpha]  
        = v^{\alpha -1} \sum_{i=1}^v \tr[Y_i^\alpha] \, . \label{eq_LS_pinch1}
    \end{equation}
    Now define $Z_i := X^{1/2} P_i$. We then see that $Z_i Z_i^* = Y_i$ and $Z_i^*Z_i = P_i X P_i$ and therefore that both $Y_i$ and $P_i X P_i$ have the same non-zero eigenvalues. Thus by orthogonality
    \begin{equation}
        \sum_{i=1}^v \tr[Y_i^\alpha] 
        = \sum_{i=1}^v \tr[(P_i X P_i)^\alpha] 
        = \tr \left[ \left( \sum_{i=1}^v P_i X P_i \right)^\alpha \right] 
        = \tr[ \cP_\sigma(X)^\alpha] \, . \label{eq_LS_pinch2}
    \end{equation}
    For $\alpha >1$, we then have 
    \begin{align}
        D_\alpha(\rho \| \sigma) 
        &= \frac{1}{\alpha - 1} \log \tr[X^\alpha] \\
        \overset{\textnormal{\Cshref{eq_LS_pinch1,eq_LS_pinch2}}}&{\leq}  \frac{1}{\alpha-1} \log \tr\left[ \cP_\sigma(X)^\alpha \right] + \log v \\
        &= D_\alpha(P_\sigma(\rho)\| \sigma) + \log v \\
        &= D_\alpha(P_\sigma(\rho)\| P_\sigma(\sigma)) + \log v \\
        &\leq D_{\alpha,\mathbb{M}}(\rho\|\sigma) + \log v \, ,
    \end{align}
    where in the penultimate step we used that $P_{\sigma}(\sigma)=\sigma$.
    Sending $\alpha \to 1$ then shows the complete statement~\cite{PH91,MLDSFT13}.
\end{proof}

\begin{lemma}\label{lem:polyrank}
    Let $A>0$ and $\rho =\ket{\psi}\bra{\psi}$ be a pure state. For $m \in \N$ and $s\in[0,1)$, define $\sigma_m := f_{s,A^{\otimes m}}(\rho^{\otimes m})$. Then 
    \begin{equation}
        \rank \, \sigma_m \leq \binom{m+\dim(A)-1}{\dim(A)-1} \leq (m+1)^{\dim(A)-1} \, .
    \end{equation}
\end{lemma}
\begin{proof}
    By~\cref{eq:f_s-integral}, we have
    \begin{equation}
        \sigma_m = \int_0^\infty h_s(t) \ket{\varphi_t}\bra{\varphi_t} \, \di t \qquad \textnormal{with} \qquad \ket{\varphi_t} = (A^{\otimes m}+t\id)^{-1} \ket{\psi}^{\otimes m} \, .
    \end{equation}
    Then writing $A^{\otimes m}$ in its spectral decomposition gives $A^{\otimes m} = \sum_a a P_a$, and thus $\ket{\varphi_t} = \sum_a \frac{1}{a + t}P_a \ket{\psi}^{\otimes m}$. Therefore, we have $\range \, \sigma_m \subseteq \mathrm{span}\, \{P_a \ket{\psi}^{\otimes m}\}_a$, which allows us to bound its rank by the size of the spectrum of $A^{\otimes m}$. Then~\cite[Lemma II.1]{csiszar98} gives
    \begin{equation}
        |\spec(A^{\otimes m})| \leq \binom{m+\dim(A)-1}{\dim(A)-1} \leq (m+1)^{\dim(A)-1} \, .
    \end{equation}
\end{proof}


We are equipped with all preliminaries to prove the assertion of~\cref{lem_var_formula_new}.
For $\rho=\proj{\psi}_{AR}$, $m \in \N$, $T > 0$, and $A=\cF_R^*(T)$ let $\sigma_m = f_{s,A^{\otimes m}}(\rho^{\otimes m})$ and note that for any $s \in [0,1)$
\begin{align}
D_{\frac{1}{1-s},\mathbb{M}}\big(\cE_R^{\otimes m}(\rho^{\otimes m}) \| \cF_R^{\otimes m}(\sigma_m)\big) - D_{\frac{1}{1-s},\mathbb{M}}(\rho^{\otimes m} \| \sigma_m)
\overset{\textnormal{\Cshref{eq:amort-lower}}}{\geq} q_s(T^{\otimes m},\rho^{\otimes m})
\overset{\textnormal{\Cshref{eq_def_q}}}{=} m q_s(T,\rho) \, . \label{eq_var_formula_step1}
\end{align}
Furthermore, we have
\begin{align}
&\hspace{-15mm}D_{\frac{1}{1-s}}\big(\cE_R^{\otimes m}(\rho^{\otimes m}) \| \cF_R^{\otimes m}(\sigma_m)\big) - D_{\frac{1}{1-s}}(\rho^{\otimes m} \| \sigma_m) \nonumber \\
\overset{\textnormal{\Cshref{lem:pinch,lem:polyrank}}}&{\geq} D_{\frac{1}{1-s},\mathbb{M}}\big(\cE_R^{\otimes m}(\rho^{\otimes m}) \| \cF_R^{\otimes m}(\sigma_m)\big) - D_{\frac{1}{1-s},\mathbb{M}}(\rho^{\otimes m} \| \sigma_m) - o(m) \\
\overset{\textnormal{\Cshref{eq_var_formula_step1}}}&{\geq}  m q_s(T,\rho) - o(m) \, . \label{eq_var_formula_step2}
\end{align}
This can be rewritten as
\begin{align}
q_s(T,\rho) 
\overset{\textnormal{\Cshref{eq_var_formula_step2} }}&{\leq}  \frac{1}{m}\Big(D_{\frac{1}{1-s}}\big(\cE_R^{\otimes m}(\rho^{\otimes m}) \| \cF_R^{\otimes m}(\sigma_m)\big) - D_{\frac{1}{1-s}}(\rho^{\otimes m} \| \sigma_m) \Big) + o(1) \\
\overset{\textnormal{\Cshref{eq_chain_rule}}}&{\leq} D_{\frac{1}{1-s}}^{\infty}(\cE \| \cF) + o(1) \, . \label{eq_var_formula_step3}
\end{align}
 Since~\cref{eq_var_formula_step3} is true for any $T>0$ and $\rho=\proj{\psi}_{AR}$ we find by taking the limit $m \to \infty$
\begin{align} \label{eq_var_formula_proof_done1}
C_s(\cE , \cF) \leq  D_{\frac{1}{1-s}}^{\infty}(\cE \| \cF) \, .
\end{align}

To see the other direction, let $(\rho_k)_{k \in \N}$ be an optimizing sequence of states. Then define the density matrices $\omega_k:=\cE_R^{\otimes k}(\rho_k)$ and $\tau_k:=\cF_R^{\otimes k}(\rho_k)$ such that for $m \in \N$ we get
\begin{align}
    D_{\frac{1}{1-s},\mathbb{M}}\big(\omega_k^{\otimes m}  \| \tau_k^{\otimes m}\big)
    &=D_{\frac{1}{1-s},\mathbb{M}}\big(\cE_R^{\otimes km}(\rho_k^{\otimes m})  \| \cF_R^{\otimes km}(\rho_k^{\otimes m})\big) \\
    \overset{\textnormal{\Cshref{lem:measured-amortization}}}&{\leq} C_s(\cE^{\otimes km} , \cF^{\otimes km})  \\
    \overset{\textnormal{\Cshref{lem:properties}}}&{=} km C_s(\cE , \cF) \, . \label{eq_var_formula_step4}
\end{align}
Hence, by definition of the sequence $(\rho_k)_{k \in \N}$ we have
\begin{align}
D_{\frac{1}{1-s}}^{\infty}(\cE \| \cF)
&=\lim_{k \to \infty} \frac{1}{k} D_{\frac{1}{1-s}} \big(\cE_R^{\otimes k}(\rho_k) \| \cF_R^{\otimes k}(\rho_k) \big) \\
&=\lim_{k \to \infty}  \frac{1}{k} D_{\frac{1}{1-s}} (\omega_k \| \tau_k)\\
&=\lim_{k \to \infty} \lim_{m \to \infty}  \frac{1}{m k} D_{\frac{1}{1-s}} \big(\omega_k^{\otimes m} \| \tau_k^{\otimes m}\big) \\
\overset{\textnormal{\cite[Fact B.3]{MSR25}}}&{=} \lim_{k \to \infty} \lim_{m \to \infty}  \frac{1}{m k} D_{\frac{1}{1-s},\mathbb{M}} \big(\omega_k^{\otimes m} \| \tau_k^{\otimes m}\big) \\
\overset{\textnormal{\Cshref{eq_var_formula_step4}}}&{\leq}C_s(\cE , \cF) \, ,
\end{align}
which together with~\cref{eq_var_formula_proof_done1} completes the proof. \qed

\subsection{Proof of~\cref{lem_cont_Cs}} \label{sec_proof_confinement}
Proving~\cref{lem_cont_Cs} consists of two main challenges. The first part is to show that the function $q_s(T,\rho)$ can be bounded from above by $C_0(\cE,\cF)$ plus a controlled error term. This is possible if we confine the support of $\rho$ to a short interval of $\log \cF_R^*(T)$.
The second challenge is to  construct a pair of near-optimal $T$ and $\rho$ that fulfill this support condition.
\begin{definition}\label{def:confined}
    We say a Hermitian matrix $Z$ is confined to width $L \geq 0$ for a Hermitian matrix $Y$ if for some $x\in \mathbb{R}$
    \begin{equation}
    \supp\,Z \subseteq \range \mathbf{1}_{[x,x+L]}(Y)\, 
\end{equation}
\end{definition}

\begin{lemma}\label{lem:newtest}
Let $\cF \in \CPTP(A,B)$, $s \in (0,1)$, $T>0$, $A=\cF_R^*(T)$, and let $\rho \in \St(RA)$ be confined to width $L\geq 0$ for $\log A$. Then there exists a $\lambda>0$ such that $S=\lambda T+\id$ satisfies
\begin{equation}
\tr[\rho\cF_R^*(S)]\leq \ee^L+1 \qquad \textnormal{and} \qquad   q_s(T,\rho)\leq q_s(S,\rho)+\log2 \, .
\end{equation}
\end{lemma}
\begin{proof}
Take $x$ from~\cref{def:confined}, and set $\lambda= \ee^{-x}$. By unitality, we get $\cF_R^*(S) = \ee^{-x} A + \id$ and since $\rho$ is supported on the interval of $A$ corresponding to the eigenvalues $[\ee^{x},\ee^{x+L}]$, we find 
\begin{equation}
    \tr[\rho \cF_R^*(S)] = \tr[\rho (\ee^{-x}A + \id)]\leq \ee^L + 1 \, .
\end{equation}
Next, we show that $q_s$ does not change too much.
By operator monotonicity of $x \to x^s$~\cite[Theorem~V.1.9]{bhatia_book}, we have 
\begin{align} \label{eq_dsdsds}
\tr[\cE_R(\rho)S^s] \geq \ee^{-xs} \, \tr[\cE_R(\rho)T^s] \, .
\end{align}
Since $\rho$ is confined for $\log A$, the  scalar inequality $(x+1)^s \leq 2^s x^s$ for $x\geq1$ implies
\begin{align}  \label{eq_dsdsds2}
    \tr[\rho\cF_R^*(S)^s] =\tr[\rho(\ee^{-x}A+\id)^s] \leq 2^s \ee^{-xs}\tr[\rho A^s] \, .
\end{align}
Together, this gives
\begin{align}
    q_s(S,\rho) 
    \overset{\textnormal{\Cshref{eq_def_q}}}{=} \frac{1}{s} \log \frac{\tr[\cE_R(\rho)S^s]}{\tr[\rho\cF_R^*(S)^s]}
    \overset{\textnormal{\Cshref{eq_dsdsds,eq_dsdsds2}}}{\geq} \frac{1}{s} \log \frac{\ee^{-xs}\tr[\cE_R(\rho)T^s]}{2^s \ee^{-xs}\tr[\rho A^s]} 
    \overset{\textnormal{\Cshref{eq_def_q}}}{=} q_s(T,\rho) - \log 2 \, .
\end{align}
\end{proof}

\begin{lemma}\label{lem:upperboundc0}
Let $\cE,\cF \in \CPTP(A,B)$, $s\in(0,1)$, $\rho \in\St(RA)$, and suppose $S \geq \id$ and $\tr[\rho\cF_R^*(S)]\leq M$. By setting $l = 1 +  \max\{ D_{\max}(\cE \| \cF) +\log M,0\} $, we get
\begin{equation}
    q_s(S,\rho)\leq C_0 + \frac{sl^2}{2(1-s)}\,.
\end{equation}
\end{lemma}
\begin{proof}

By Jensen's inequality, we have $\frac{1}{s}\log \tr[\rho \cF_R^*(S)^s]\geq \tr[\rho \log \cF_R^*(S)]$. This implies
\begin{align}
    q_s(S,\rho) - q_0(S,\rho) 
    \overset{\textnormal{\Cshref{eq_def_q}}}&{=} \frac{1}{s} \log\frac{\tr[\cE_R(\rho)S^s]}{\tr[\rho \cF_R^*(S)^s]} - \tr[\cE_R(\rho)\log S] + \tr[\rho \log \cF_R^*(S)] \\
    &\leq \frac{1}{s}\log \tr[\cE_R(\rho)S^s] - \tr[\cE_R(\rho)\log S] \,. \label{eq_LS_1}
\end{align}
Now define $F(t) := \log \tr[\cE_R(\rho)S^t]$ for $t \in [0,1]$. For this function we have $F(0)=0$ and $F'(0) = \tr[\cE_R(\rho)\log S]$.
Furthermore,
\begin{equation}
    \ee^{F(1)} = \tr[\cE_R(\rho)S] \leq \ee^{D_{\max}(\cE \| \cF)} \tr[\rho \cF^*_R(S)] \leq M \ee^{D_{\max}(\cE \| \cF)}
\end{equation}
and therefore $F(1)\leq l-1 $. Calculating the second derivative gives
\begin{equation}
    F''(t) 
    = \frac{\tr[\cE_R(\rho)S^t(\log S)^2]}{\tr[\cE_R(\rho)S^t]}-F'(t)^2 
    \leq \frac{\tr[\cE_R(\rho)S^t(\log S)^2]}{\tr[\cE_R(\rho)S^t]} \label{eq_LSLS2}
\end{equation}
This can be further bounded by
\begin{align}
    \frac{\tr[\cE_R(\rho)S^t(\log S)^2]}{\tr[\cE_R(\rho)S^t]} &\leq \frac{1}{(1-t)^2}\frac{\tr[\cE_R(\rho)S^t(\id+(1-t)\log S)^2]}{\tr[\cE_R(\rho)S^t]}  \label{eq_ds_FIRST}\\
    &\leq \frac{1}{(1-t)^2}\left( 1+ \log \frac{\tr [\cE_R(\rho)S]}{\tr[\cE_R(\rho)S^t]}\right)^2 \label{eq_ds_SEC} \\
    & = \frac{1}{(1-t)^2} (1+ F(1)-F(t))^2 \label{eq_ds_THIRD}\\
    &\leq \frac{l^2}{(1-t)^2}  \, . \label{eq_LSLS3}
\end{align}
\Cref{eq_ds_FIRST} is due to
\begin{equation}
    (\id +(1-t)\log S)^2 - (1-t)^2 (\log S)^2 = \id +2(1-t)\log S \geq 0 \, .
\end{equation}
If we multiply this expression by $S^t$ and afterwards take the trace against $\cE_R(\rho)$, we get~\cref{eq_ds_FIRST}.
\Cref{eq_ds_SEC} follows from applying Jensen's inequality to the  function $f(x)= (1+\log x)^2$, which is concave for $x\geq1$. 
Since $S\geq \id$, working in an eigenbasis of $S$ with eigenvalues $\lambda_i$ leads to
\begin{equation}
     \frac{\tr [\cE_R(\rho)S]}{\tr[\cE_R(\rho)S^t]} = \sum_i \lambda_i^{1-t} \frac{\lambda_i^t\bra{i}\cE_R(\rho)\ket{i}}{\tr[\cE_R(\rho)S^t]} = \sum_i\lambda_i^{1-t} p_i \, ,
\end{equation}
where we introduced the probability distribution $p_i$. Applying Jensen's inequality to $f(\sum_i p_i \lambda_i^{1-t})$ explains~\cref{eq_ds_SEC}. \Cref{eq_ds_THIRD,eq_LSLS3} follow by definition and $0\leq F(t) \leq F(1) \leq l-1$.

Now fully equipped, we can show
\begin{align}
     q_s(S,\rho) - q_0(S,\rho) 
    \overset{\textnormal{\Cshref{eq_LS_1}}}&{\leq} \frac{1}{s}\log \tr[\cE_R(\rho)S^s] - \tr[\cE_R(\rho)\log S] \\
    &= \frac{1}{s}F(s)-F'(0) \\
    &= \frac{1}{s} \int_0^s (s-t)F''(t)\, \di t \label{eq_ds_unclear} \\
    \overset{\textnormal{\Cshref{eq_LSLS2,eq_LSLS3}}}&{\leq} \frac{l^2}{s} \int_0^s \frac{s-t}{(1-t)^2} \, \di t \\
    &\leq \frac{sl^2}{2(1-s)} \, ,
\end{align}
where~\cref{eq_ds_unclear} follows from $F(0)=0$ and
\begin{align}
    F(s)-F(0) - sF'(0) &= \int_0^s (F'(u)-F'(0)) \, \di u \\
    &= \int_0^s \int_0^u F''(t) \, \di t \, \di u\\
    &= \int_0^s \int_t^s F''(t) \, \di u \, \di t\\
    &= \int_0^s (s-t) F''(t) \, \di t \, .
\end{align}
Together with $q_0(S,\rho) \leq C_0(\cE,\cF)$, this concludes the proof. 
\end{proof}

\begin{corollary}\label{cor:upper}
Let $\cE,\cF \in \CPTP(A,B)$, $s\in (0,1)$, $K=D_{\max}(\cE\|\cF) +3$, and $n \in \N$. For finite reference $R$ and test $T>0$, let $\rho$ be a state that is confined to width $L \in [0,n]$ for $\log( (\cF_R^*)^{\otimes n} (T))$. Then
    \begin{equation}
        q_s(T,\rho) \leq nC_0(\cE,\cF) + \log 2 + \frac{sK^2n^2}{2(1-s)} \, ,
    \end{equation}
    where $q_s$ corresponds to $(\cE^{\otimes n},\cF^{\otimes n})$.
\end{corollary}
\begin{proof}
    By~\cref{lem:newtest}, we can replace the given test $T$ with a new test $S$ if the input state is confined. Then we have $q_s(T,\rho) \leq q_s(S,\rho) + \log 2$. Since this test satisfies $S\geq \id$ and we also have $D_{\max}(\cE^{\otimes n}\|\cF^{\otimes n}) = n D_{\max}(\cE\|\cF)$, and $\tr[\rho (\cF^*_R)^{\otimes n}(S)]\leq \ee^L + 1$ by~\cref{lem:newtest}, we can use~\cref{lem:upperboundc0} to obtain 
    \begin{equation}
        q_s(S,\rho) \leq C_0(\cE^{\otimes n},\cF^{\otimes n}) + \frac{sl^2}{2(1-s)}
    \end{equation}
    with 
    \begin{align}
  l 
  = 1 + n D_{\max}(\cE\|\cF) + \log (\ee^L+1) 
  \leq 2 + n D_{\max}(\cE\|\cF) + L \leq n(D_{\max}(\cE\|\cF) + 3) 
  = nK \, .
    \end{align}
     The assertion of~\cref{cor:upper} then follows from additivity of $C_0(\cE^{\otimes n},\cF^{\otimes n})$ (see~\cref{lem:properties}).
\end{proof}

We continue by relating a nearly optimal witness at order $s/2$ to $s$.

\begin{lemma}\label{lem:orderineq}
Let $\cE,\cF \in \CPTP(A,B)$, $s\in(0,1)$, $T>0$, $A=\cF_R^*(T)$, $N_s(T)=A^{-s/2}\cE_R^*(T^s)A^{-s/2}$, choose $\lambda \geq \norm{N_s(T)}_\infty $ and set $H=\id -\lambda^{-1}N_s(T)$. Then 
\begin{equation}\label{eq:lemineq}
    \id - \lambda^{-1/2}N_{s/2}(T) \geq \frac{1}{2} \int_\mathbb{R} \mu_s(t) A^{\ci t} H A^{-\ci t} \, \di t \, ,
\end{equation}
where
\begin{equation}
    \mu_s(t) = \frac{2}{s \cosh(2\pi t/s)} \, 
\end{equation}
is a probability density.
\end{lemma}
\begin{proof}
    Define the matrix $B = A^{s/2} - \lambda^{-1/2} \sqrt{\cE_R^*(T^s)}$. Using that $t \mapsto \sqrt{t}$ is operator concave~\cite{bhatia_book} together with Choi-Davis-Jensen's inequality~\cite{choi74,davis57}, we have $\cE_R^*(T^{s/2})\leq \sqrt{\cE_R^*(T^s)}$. Therefore, conjugating $B$ by $A^{-s/4}$ and using the definition of $N_s(T)$ gives
    \begin{equation}\label{eq:unitjensen}
        \id - \lambda^{-1/2} N_{s/2}(T) \geq A^{-s/4} B A^{-s/4} \, .
    \end{equation}
    By definition of $B$ and $H$ we further get
    \begin{equation}
        (A^{s/2} - B)^2 = \lambda^{-1} \cE_R^*(T^s) = A^{s/2}(\id-H)A^{s/2}
    \end{equation}
    or equivalently
    \begin{equation}\label{eq:id1}
        A^{s/2} B + B A^{s/2} = A^{s/2}HA^{s/2} + B^2 \, .
    \end{equation}

    Now consider 
    \begin{equation}
        \frac{\di}{\di t} \left(\ee^{-tA^{s/2}} B \ee^{-tA^{s/2}}\right) = - \ee^{-tA^{s/2}} A^{s/2} B \ee^{-tA^{s/2}} -  \ee^{-tA^{s/2}} B A^{s/2} \ee^{-tA^{s/2}} \, . \label{eq_LSLSLS}
    \end{equation}
    Conjugating~\cref{eq:id1} with $\ee^{-tA^{s/2}}$, dropping $B^2$ and then integrating over $t\geq 0$ leads to
    \begin{equation}
        B \overset{\textnormal{\Cshref{eq_LSLSLS}}}{\geq} \int_0^\infty \ee^{-tA^{s/2}} A^{s/2}HA^{s/2}  \ee^{-tA^{s/2}} \, \di t \, .
    \end{equation}
    Another conjugation with $ A^{-s/4}$ then gives a lower bound to $\id - \lambda^{-1/2}N_{s/2}(T)$ by~\cref{eq:unitjensen}. 
    Now it remains to bring the lower bound into its final form. Working in an eigenbasis of $A=\diag(a_i)$ gives
    \begin{equation}
        \left(\int_0^\infty \ee^{-tA^{s/2}} A^{s/4} H A^{s/4}  \ee^{-tA^{s/2}} \, \di t \right)_{ij} = a_i^{s/4} a_j^{s/4} H_{ij} \int_0^\infty \ee^{-t(a_i^{s/2}+a_j^{s/2})} \, \di t= \frac{H_{ij}}{2\cosh\left(\frac{s}{4}\log \frac{a_i}{a_j}\right)} \, .
    \end{equation}
    Using the identity (see~\cite[Equation~(3.982.1)]{gradshteyn_ryzhik})
    \begin{equation}
        \frac{1}{\cosh(sb/4)} = \int_\mathbb{R} \mu_s(t) \ee^{\ci tb} \, \di t, \qquad \forall b\in \mathbb{R} \, ,
    \end{equation}
    together with $(A^{it}HA^{-it})_{ij} = \ee^{\ci t \log (a_i/a_j)} H_{ij}$ proves~\cref{eq:lemineq}.
\end{proof}

\begin{lemma}\label{lem:average}
For any Hermitian matrix $Y$, unit vector $\ket{v}$, $0 \leq H \leq \id$ and $L>0$, define $P_x = \mathbf{1}_{[x,x+L)}(Y)$.
Then, there exists an $x$ with $P_x \ket{v} \neq 0$ such that $\ket{v_x}=\frac{P_x \ket{v}}{\norm{P_x \ket{v}}}$ satisfies
\begin{equation}
    \langle v_x, Hv_x \rangle \leq \int_\mathbb{R} \omega_L(t) \langle v, \ee^{\ci tY} H \ee^{-\ci tY} v\rangle \, \di t \, ,
\end{equation}
where
\begin{equation}
    \omega_L(t) = \frac{2 \sin^2(Lt/2)}{\pi L t^2} 
\end{equation}
is a probability density.
\end{lemma}
\begin{proof}
    Using the eigenbasis of $Y$, each eigenvalue $y_i$ belongs to spectral windows of length $L$. Two eigenvalues $y_i,y_j$ can belong to the same window if their distance is less than $L$. Thus $(P_x H P_x)_{ij} = (P_x)_{ii} H_{ij} (P_x)_{jj} =  H_{ij}$ if $y_i,y_j \in [x,x+L)$ and zero otherwise.
    Counting the number of windows containing both then gives a set of length $\max\{L-|y_i-y_j|,0\}$ and therefore 
    \begin{equation}
        \frac{1}{L}\int (P_x H P_x)_{ij} \, \di x = \max \left\{ 1-\frac{|y_i-y_j|}{L},0\right\}H_{ij} \, .
    \end{equation}
    If we then use the identity~\cite[Section~2.1]{fourier_book}
    \begin{equation}
        \int_\mathbb{R} \omega_L(t) \ee^{\ci ta} \, \di t = \max\left \lbrace 1-\frac{|a|}{L},0\right \rbrace \quad \forall a \in \R \, 
    \end{equation}
    together with $(\ee^{\ci tY}H \ee^{-\ci tY})_{ij} = \ee^{\ci t (y_i -y_j)} H_{ij}$, we get the operator identity
    \begin{equation}
        \frac{1}{L} \int_\mathbb{R} P_x H P_x \, \di x = \int_\mathbb{R} \omega_L(t)  \ee^{\ci tY}H \ee^{-\ci tY} \, \di t \, .
    \end{equation}
    If we now take the expectation with $v$, we get $\langle P_x v, H P_x v\rangle = \norm{P_x v}^2 \langle v_x,H v_x \rangle$ on the left side. Since $1/L \int_\mathbb{R} P_x  \di x= \id$ and $\norm{v}=1$, we have $1/L \int_\mathbb{R} \norm{P_x v}^2 \, \di x = 1$. So the left-hand side is an average of $\langle v_x,H v_x \rangle$ with weights $\norm{P_x v}^2/L$. In particular, there has to be some $x$ with $P_x v \neq 0$, such that $\langle v_x,H v_x \rangle$ is not larger than this average. This concludes the proof.
\end{proof}

\begin{lemma}\label{lem:abscont}
    Let $\kappa$ and $\omega$ be probability densities on $\mathbb{R}$ with $\kappa>0$ almost everywhere. For measurable functions $f_s: \mathbb{R} \to [0,1]$, 
    \begin{equation}
     \lim_{s\downarrow0}  \int \kappa(t) f_s(t) \, \di t = 0  
\quad \implies \quad 
        \lim_{s\downarrow0}  \int \omega(t) f_s(t) \, \di t = 0 \, .
    \end{equation}
\end{lemma}
\begin{proof}
    This is a standard consequence of absolute continuity of finite measures (see for instance~\cite[Lemma~5.21]{Salamon}).
\end{proof}

\begin{lemma}\label{lem:test}
    Let $\cE,\cF \in \CPTP(A,B)$ and $\varepsilon>0$. Furthermore, set $L = \varepsilon/s$ and $n= \lceil L \rceil$. Then for sufficiently small $s$, there are $d_s\in[0,1)$ with $d_s \to 0$ as $s \downarrow 0$, such that there are a finite reference $R$, a positive definite test $T$ and a pure state $\rho_s$ that is confined to width $L$ for $\log (\cF_R^*)^{\otimes n} (T)$ that satisfies
    \begin{equation}
        q_s(T,\rho_s) \geq nC_s(\cE,\cF) + \frac{1}{s} \log(1-d_s) \, .
    \end{equation}
\end{lemma}
\begin{proof}
    By~\cref{lem:properties}, the right limit $s\downarrow0$ of $C_s(\cE,\cF)$ exists and $\delta_s \coloneqq C_s(\cE,\cF) - C_{s/2}(\cE,\cF) +s$ vanishes for $(s\downarrow0)$. Furthermore, we have $\varepsilon\leq ns \leq \varepsilon + s$ and therefore $ns \to \varepsilon$ as $s \downarrow 0$. By additivity and because the supremum may be restricted to pure states (see~\cref{lem:properties}), we can find a finite reference, $T>0$ and a pure $\rho$ such that 
    \begin{equation}\label{eq:supb}
        q_{s/2}(T,\rho)\geq nC_{s/2}(\cE,\cF) - ns \, .
    \end{equation} Let us now define $A =  (\cF_R^*)^{\otimes n} (T)$, $\lambda = \ee^{snC_s(\cE,\cF)}$ and $H=\id - \lambda^{-1} N_s(T)$. By~\cref{lem:properties}, we have $N_s(T)\leq \lambda \id$ and therefore $0\leq H \leq \id$. Now, let \smash{$\ket{v}\bra{v} = \frac{A^{s/4}\rho A^{s/4}}{\tr[\rho A^{s/2}]}$} and $g_s(t) = \langle v, A^{\ci t}HA^{-\ci t} v \rangle$. By sandwiching the inequality in~\cref{lem:orderineq} with $v$, we get
    \begin{align}
        \int_\mathbb{R} \mu_s(t) g_s(t) \, \di t 
        \overset{\textnormal{\Cshref{lem:orderineq}}}&{\leq} 2 \left(1- \lambda^{-1/2}\langle v,N_{s/2}(T)v\rangle\right) \\
        \overset{\textnormal{\Cshref{eq:tilt}}}&{=} 2 \left(1-\ee^{-\frac{s}{2}nC_s(\cE,\cF)}\ee^{\frac{s}{2}q_{s/2}(T,\rho)}\right) \\
        \overset{\textnormal{\Cshref{eq:supb}}}&{\leq} 2 \left(1-\ee^{-\frac{sn}{2}(C_s(\cE,\cF) - C_{s/2}(\cE,\cF) + s)} \right) \\
        &\leq sn \delta_s \,,
    \end{align}
    where the last inequality uses $1-\ee^{-x} \leq x$.  Note that $\lim_{s\downarrow 0}sn \delta_s = 0$.
    To get a measure that is independent of $s$, note that the change of variables from $t$ to $su$ gives
    \begin{equation}
        \int_\mathbb{R} \mu_s(t) g_s(t) \, \di t =  \int_\mathbb{R} s\mu_s(us) g_s(us) \, \di u = \int_\mathbb{R} \kappa(u) h_s(u) \, \di u \, ,
    \end{equation}
    where we introduced the function $h_s(u) \coloneqq g_s(us)$ and the density $\kappa(u) \coloneqq s \mu_s(su) = \frac{2}{\cosh(2\pi u)} > 0$. Analogously, we get for the density in~\cref{lem:average} 
\begin{align}
\nu_\varepsilon(u) \coloneqq s\omega_{L}(su) 
= s\frac{2 \sin^2(L su/2)}{\pi Ls^2 u^2} \, = \frac{2 \sin^2(\varepsilon u/2)}{\pi \varepsilon u^2} \, ,
\end{align}
    which for fixed $\varepsilon$ is independent of $s$. According to~\cref{lem:abscont}, $\lim_{s \downarrow 0}\int_\mathbb{R} \kappa(u) h_s(u) \, \di u =  0$ implies  $\lim_{s \downarrow 0}\int_\mathbb{R} \nu_\varepsilon(u) h_s(u) \, \di u  = 0$. Thus let us define
    \begin{equation}
        d_s \coloneqq \int_\mathbb{R} \nu_\varepsilon(u) h_s(u) \, \di u = \int_\mathbb{R} \omega_{L}(t) g_s(t) \, \di t \, .
    \end{equation}
    Applying~\cref{lem:average} to $Y = \log A$ and $L = \varepsilon/s$, it gives a normalized projected vector $v_x$ supported in $[x,x+L)$ of the spectrum of $\log A$ such that $\langle v_x, H v_x \rangle \leq d_s$. We can now turn $v_x$ into an input state $\rho_s$ by 
    \begin{equation}
        \rho_s = \frac{A^{-\frac{s}{2}}\ket{v_x}\bra{v_x}A^{-\frac{s}{2}}}{\langle v_x, A^{-s} v_x \rangle} \, .
    \end{equation}
    This state is normalized, pure and conjugating by $A^{-\frac{s}{2}}$ preserves the spectral subspace of $\log A$. So $\rho_s$ is still confined to width $L$ for $\log A$. Applying~\cref{eq:tilt} to $\rho_s$ gives 
    \begin{equation}
        \ee^{sq_s(T,\rho_s)} = \langle v_x, N_s(T) v_x\rangle = \lambda(1-\langle v_x, H v_x\rangle) \geq \ee^{snC_s(\cE,\cF)} (1-d_s) \, .
    \end{equation}
\end{proof}

\begin{proof}[Proof of~\cref{lem_cont_Cs}]
Without loss of generality we can assume that $\cE$ and $\cF$ satisfy $D_{\max}(\cE\|\cF)<\infty$ as otherwise the assertion of~\cref{lem_cont_Cs} clearly holds.
For fixed $\varepsilon>0$, $L = \varepsilon/s$, $n = \lceil L \rceil$ and sufficiently small $s$,~\cref{lem:test} provides a test $T$, a confined state $\rho_s$ and $d_s \to 0$ for $s\downarrow 0$ satisfying
    \begin{equation} \label{eq_LS_final1}
        q_s(T,\rho_s) \geq n C_s(\cE,\cF) + \frac{1}{s} \log(1-d_s) \, .
    \end{equation}
    On the other hand,~\cref{cor:upper} guarantees for such $(T,\rho_s)$ that
    \begin{equation}\label{eq_LS_final2}
        q_s(T,\rho_s) \leq nC_0(\cE,\cF) + \log 2 + \frac{s K^2 n^2}{2(1-s)} \, ,
    \end{equation}  
    where $K = D_{\max}(\cE\|\cF) + 3$. Rearranging and dividing by $n$ then gives
    \begin{equation}
        0 
        \overset{\textnormal{\Cshref{lem:properties}}}{\leq} C_s(\cE,\cF) - C_0(\cE,\cF) 
        \overset{\textnormal{\Cshref{eq_LS_final1,eq_LS_final2}}}{\leq} \frac{\log 2}{n} + \frac{K^2}{2(1-s)}sn - \frac{\log(1-d_s)}{sn} \, .
    \end{equation}
Now let $s \downarrow 0 $
with $\varepsilon$ fixed. Then the first and the last terms on the right hand side tend to zero because $n \to \infty$, $d_s \to 0$ and $sn \to \varepsilon$. For $s<1/2$, we can bound the leftover term and get
\begin{equation}
    \limsup_{s \downarrow 0} C_s(\cE,\cF) - C_0(\cE,\cF) \leq  K^2 \varepsilon \, .
\end{equation}
Since this works for every $\varepsilon > 0$, this proves $\lim_{s \downarrow 0} C_s(\cE,\cF) = C_0(\cE,\cF)$.
\end{proof} 

\subsection{It suffices that the second channel is completely positive} \label{sec_CP_sufficient}
As explained in~\cref{eq_overview_proof} the assertion of~\cref{thm:main} for CPTP maps $\cE$ and $\cF$ follows by combining~\cref{lem_var_formula_new,lem_cont_Cs}.
We next show that the statement can be lifted to $\cF \in \CP(A,B)$. Choose $c>0$ such that
$\cF^*(\id_B)\le c \id_A$ and define channels into
$B\oplus \mathbb{C}$ by
\begin{equation}
    \cE'(X):=\cE(X) \oplus 0 \qquad \textnormal{and} \qquad \cF'(X) := c^{-1} \cF(X)\oplus \bigl(\tr[X]-c^{-1}\tr[\cF(X)]\bigr)\, .
\end{equation}
Then we have $\cE',\cF' \in \CPTP(A,B\oplus \mathbb{C})$.
At every tensor power, the output of $(\cE')^{\otimes n}$
is supported on $B^{\otimes n}$, where the corresponding block of $(\cF')^{\otimes n}$ equals $c^{-n} \cF^{\otimes n}$.
Therefore, we get for every $\alpha>1$
\begin{equation}
    D_\alpha^\infty(\cE'\|\cF') = D_\alpha^\infty(\cE\|\cF)+\log c\, ,
\end{equation}
and the same identity holds for $D^\infty$. Applying~\cref{eq:main} to $\cE',\cF'$ and subtracting $\log c$ from both sides completes the argument.

\section{Applications} \label{sec_applications}

\subsection{Computability of channel divergence} \label{sec_computability}
Due to the regularization in the formula, it is unclear if the regularized channel divergence is computable~\cite{FFF25_2}. First hints that this may be the case have been obtained in~\cite[Theorem~5.1]{FF21} where it was shown that the regularized R\'enyi channel divergence for $\alpha>1$ is computable. 
\Cref{thm:main} allows us to prove that this is also true for the case $\alpha =1$.
\begin{corollary} \label{cor_computable}
Let $\cE,\cF \in \CPTP(A,B)$. Then the regularized channel divergence $D^\infty(\cE\|\cF)$ is computable.
\end{corollary} 
\begin{proof}
The identity $D^\infty(\cE\|\cF) = C_0(\cE,\cF)$ from~\cref{lem_var_formula_new} allows us to calculate better and better lower bounds by using finite reference systems, positive tests and density matrices. \cite[Theorem 5.1]{FF21} allows us to calculate upper bounds for $D_{1+1/k}^\infty$, for any $k \in \N$. \Cref{thm:main} then implies that these bounds approach $D^\infty(\cE\|\cF)$ arbitrarily closely (by making $k$ sufficiently large). Once the difference is below the desired precision, we terminate.
\end{proof}

\subsection{Quantum channel Stein's lemma} \label{sec_channel_disc}
In~\cite{FF21,FGX25} it was observed that the continuity of the regularized channel R\'enyi divergence gives rise to a strong converse for channel discrimination. 
For a fixed strategy $T$, let $\rho_n$ and $\sigma_n$ be the final states before the measurement and let $\{P_n,\id-P_n\}$ denote the final binary POVM, where $P_n$ corresponds to the decision $\cE$. 
The type-I and type-II errors $ \alpha_n(T,\cE) := \tr[(\id-P_n)\rho_n]$ and $\beta_n(T,\cF):=\tr[P_n\sigma_n]$ are the probabilities of accepting $\cF$ under the assumption $\cE$ and accepting $\cE$ under $\cF$ respectively.
Let us define
\begin{equation}
    \beta_{n,x}^\varepsilon(\cE \| \cF) \coloneqq \inf_{T\in \cT_{n,x}} \{\beta_n(T,\cF) : \alpha_n(T,\cE) \leq \varepsilon\}, \qquad x\in\{\textnormal{parallel, adaptive}\} \, ,
\end{equation}
where $\cT_{n,x}$ denotes the set of parallel or adaptive discrimination strategies.  

\begin{corollary}\label{cor:stein}
    Let $\cE,\cF \in \CPTP(A,B)$ and $\varepsilon \in (0,1)$. Then, 
    \begin{equation}\label{eq:channel-stein}
        \lim_{n \to \infty} -\frac{1}{n} \log \beta_{n,\textnormal{parallel}}^\varepsilon(\cE\|\cF) = \lim_{n \to \infty} -\frac{1}{n} \log \beta_{n,\textnormal{adaptive}}^\varepsilon(\cE\|\cF) = D^\infty(\cE\|\cF) \, .
    \end{equation}
    Moreover, for every $r>D^\infty(\cE\|\cF)$ there exists $c_r >0$ such that any $T \in \cT_{n,\textnormal{adaptive}}$ with errors $\alpha_n(T,\cE), \beta_n(T,\cF)$ satisfying $\beta_n (T,\cF) \leq \ee^{-n r}$ obeys $1-\alpha_n(T,\cE) \leq \ee^{-n c_r}$. 
\end{corollary}

\begin{proof}
    For $\alpha>1$, let $\rho_n$ and $\sigma_n$ be the final states of an adaptive $n$-use strategy before its binary measurement. By the chain rule in~\cite[Corollary 5.2]{FF21}, each channel use increases the R\'enyi divergence by at most $D_\alpha^\infty(\cE\|\cF)$. Intermediate channels cannot increase it by data-processing. Since the initial states coincide, we get
    \begin{equation}
        D_\alpha(\rho_n \| \sigma_n) \leq n D_\alpha^\infty(\cE \|\cF) \, .
    \end{equation}
    The final measurement then accepts $\cE$ with probability $1-\alpha_n$ and $\beta_n$, respectively. Data-processing therefore gives
    \begin{align}
        D_\alpha(\rho_n\|\sigma_n) 
        \geq \frac{1}{\alpha-1} \log \big((1-\alpha_n)^\alpha \beta_n^{1-\alpha} + \alpha_n^\alpha(1-\beta_n)^{1-\alpha}\big)
        \geq  \frac{\alpha}{\alpha-1} \log (1-\alpha_n) - \log \beta_n \, .
    \end{align}
    Combining and rearranging thus leads to
    \begin{equation}\label{eq:channel-converse}
        1 - \alpha_n \leq \beta_n^{\frac{\alpha-1}{\alpha}}\exp\Big(\frac{n(\alpha-1)}{\alpha}D_\alpha^\infty(\cE\|\cF) \Big) \, .
    \end{equation}
    Now take $r>D^\infty(\cE\|\cF)$. \Cref{thm:main} then allows us to choose an $\alpha>1$ such that $D_\alpha^\infty(\cE \| \cF) <r$. For $\beta_n \leq \ee^{-nr}$,~\cref{eq:channel-converse} gives $1-\alpha_n \leq \ee^{-nc_r}$, with $c_r = \frac{\alpha-1}{\alpha}(r-D_\alpha^\infty(\cE\| \cF)) >0$. This then also proves the converse of~\cref{eq:channel-stein}.
    
    For achievability, let $k \geq 1$ and take an input state $\ket{\psi_k}$ on $RA^{\otimes k}$ and corresponding output states $\rho_k = \cE_R^{\otimes k}(\psi_k)$ and $\sigma_k = \cF_R^{\otimes k}(\psi_k)$. For $n = mk + t, \, t\in\{0,\dots,k-1\}$, use the input state $\ket{\psi_k}^{\otimes m}$ on the first $mk$ channel inputs and ignore the remaining $t$. The resulting states are $\rho_k^{\otimes m}$ and $\sigma_k^{\otimes m}$. 
    Applying the quantum Stein lemma for states~\cite[Theorem 2]{ogawa00} together with $m/n\to 1/k$ then gives
    \begin{equation}
        \liminf_{n\to \infty} -\frac{1}{n} \log \beta_{n,\textnormal{parallel}}^\varepsilon \geq \frac{1}{k}D(\rho_k \| \sigma_k) \, .
    \end{equation}
    Taking the supremum over $\ket{\psi_k}$ and then over $k$ gives $D^\infty(\cE \| \cF)$. Adaptively, the same rate is achievable, as every parallel strategy is also an adaptive strategy.
\end{proof}



\subsection{Channel asymptotic equipartition property} \label{sec_channel_AEP}
The AEP for states~\cite{tomamichel09,marco_book} shows that for any $\rho,\sigma \in \St(\cH)$ and any $\eps \in (0,1)$ we have
\begin{align}
\lim_{n \to \infty} \frac{1}{n} D^{\eps}_{\max}(\rho^{\otimes n} \| \sigma^{\otimes n}) = D(\rho \| \sigma) \, .
\end{align}
It was conjectured in~\cite{openQuestionAndreas} that a similar AEP holds for channels when smoothing is done with respect to the diamond norm.
Let us define
\begin{equation}
    D_{\max}^{\varepsilon,\leq}(\cE\|\cF) := \inf_{\cE' \in \cB^\varepsilon(\cE)} D_{\max}(\cE'\|\cF)
\end{equation}
where $\cB^\varepsilon(\cE)=\{ \cE' \in \CPTN : \norm{\cE-\cE'}_{\diamond,+} \leq \varepsilon\}$.
Here $\CPTN$ denotes the set of completely positive, trace-nonincreasing maps and $\norm{\cdot}_{\diamond,+}$ is the generalized diamond norm 
\begin{equation}
    \norm{\cE-\cE'}_{\diamond,+} 
    := \frac{1}{2}\sup_{\psi_{RA}} \Big\{ \bigl|\cE_R(\psi_{RA})-\cE'_R(\psi_{RA})\bigr|_1 + \bigl|\tr[\cE_R(\psi_{RA})-\cE'_R(\psi_{RA})]\bigr| \Big\}
\end{equation}
defined in~\cite{gour2026failaep}. Note that in the case of $\CPTP$ maps, we have $\norm{\cE - \cF}_{\diamond,+} = \frac{1}{2}\norm{\cE-\cF}_\diamond$.

In~\cite[Theorem 14]{gour2026failaep} it was shown that an AEP for this quantity is equivalent to a strong converse for parallel channel discrimination. Therefore~\cref{cor:stein} implies the following AEP.
\begin{corollary}[Channel AEP]
    Let $\cE,\cF \in \CPTP(A,B)$ and $\varepsilon \in (0,1)$. Then, 
    \begin{equation}
        \lim_{n\to \infty}\frac{1}{n} D_{\max}^{\varepsilon,\leq}\bigl(\cE^{\otimes n} \| \cF^{\otimes n}\bigr) = D^\infty(\cE \| \cF) \, .
    \end{equation}
\end{corollary}
\noindent \cite{gour2026failaep} shows that if one restricts the smoothing to $\CPTP$ maps an AEP fails in general.

\subsection{Relative entropy accumulation} \label{sec_REAT}
The entropy accumulation theorem (EAT)~\cite{DFR16,MetgerFawziSutterRenner2024} reduces the smooth min-conditional entropy analysis of a correlated multi-round quantum process to the conditional entropy analysis of its individual rounds.
The conditional entropy can be phrased in terms of relative entropy in the sense that $H(A|B)_{\rho}=-D(\rho_{AB}\| \id_A \otimes \rho_B)$. Hence a generalization of the EAT is a relative entropy accumulation theorem (REAT). 
For this, define
\begin{equation}
    D_{\max}^{\varepsilon,\text{out}}(\cE\|\cF) \coloneqq \sup_{\psi_{RA}}D_{\max}^\varepsilon(\cE_R(\psi_{RA}) \| \cF_R (\psi_{RA})) \, .
\end{equation}

Let $\cA_i \in \CPTP(R_{i-1}S_{i-1}, R_iS_i)$ and $\cB_i \in \CP(R_{i-1}S_{i-1}, R_iS_i)$, and $\Bar{\cA}_i = \tr_{R_i} \circ \cA_i$ and  $\Bar{\cB}_i = \tr_{R_i} \circ \cB_i$. Assume $\Bar{\cB}_i = \cT_i \circ \tr_{R_{i-1}}$ for a completely positive map $\cT_i$ and furthermore that for a fixed pair of maps $\cN \in \CPTP(X,Y)$ and $\cM \in \CP(X,Y)$ with $\D_{\max}(\cN \| \cM) < \infty$, there exist $\CPTP$ maps $\cU_i : R_{i-1}S_{i-1} \to K_i X$ and $\cV_i : K_iY \to S_i$ such that 
\begin{equation}
    \Bar{\cA_i} = \cV_i \circ (\cI_{K_i} \otimes \cN) \circ \cU_i, \qquad \Bar{\cB_i} = \cV_i \circ (\cI_{K_i} \otimes \cM) \circ \cU_i\,,
\end{equation}
where the auxiliary system $K_i$ may grow with $i$. Then define $\cP_n = \tr_{R_n} \circ \cA_n \circ \cdots \circ \cA_1$ and $\cQ_n = \tr_{R_n} \circ \cB_n \circ \cdots \circ \cB_1$.
\begin{corollary}[REAT] \label{cor_REAT}
    Under these assumptions, we get for every $\varepsilon \in (0,1)$
    \begin{equation}
        \limsup_{n \to \infty} \frac{1}{n} D_{\max}^{\varepsilon,\textnormal{out}}(\cP_n\|\cQ_n) \leq D^\infty(\cN \| \cM) \, .
    \end{equation}
\end{corollary}
\begin{proof}
    Fix $\alpha>1$, let $Z$ be a finite reference system and let $\eta \in \St(ZR_0S_0)$ be an arbitrary initial state. Then let us define recursively $\rho_0=\sigma_0 = \eta$, $\rho_i = (\cI_Z \otimes \cA_i)(\rho_{i-1})$ and $\sigma_i = (\cI_Z \otimes \cB_i)(\sigma_{i-1})$. 
    Further, define $\Bar{\rho}_i = \tr_{R_i} \rho_i$ and $\Bar{\sigma}_i = \tr_{R_i} \sigma_i$. Since $\Bar{\cB}_i = \cT_i \circ \tr_{R_{i-1}}$, the strengthened chain rule in~\cite[Theorem 3.1]{MetgerFawziSutterRenner2024} gives
    \begin{align}
        D_\alpha(\Bar{\rho_i} \| \Bar{\sigma}_i ) &= D_\alpha \bigl( (\cI_Z \otimes \Bar{\cA}_i)(\rho_{i-1}) \| (\cI_Z \otimes \Bar{\cB}_i)(\sigma_{i-1})\bigr) \\
        &\leq D_\alpha(\Bar{\rho}_{i-1}\| \Bar{\sigma}_{i-1}) + D_\alpha^\infty (\Bar{\cA}_i \| \Bar{\cB}_i) \\
        &\leq  D_\alpha(\Bar{\rho}_{i-1}\| \Bar{\sigma}_{i-1}) + D_\alpha^\infty (\cN \| \cM) \, ,
    \end{align}
    where the last step follows from data-processing. 
    Using this iteratively together with $\Bar{\rho}_0 = \Bar{\sigma}_0$ leads to
    \begin{equation}
        D_\alpha(\Bar{\rho}_n \| \Bar{\sigma}_n ) = D_\alpha \bigl( (\cI_Z \otimes \cP_n)(\eta) \| (\cI_Z \otimes \cQ_n)(\eta) \bigr) \leq n D_\alpha^\infty(\cN \| \cM) \, .
    \end{equation}
    Then the smoothing bound from~\cite[Theorem~3]{AnshuBertaJainTomamichel2019} gives
    \begin{equation}
        D^{\varepsilon}_{\max}\bigl( (\cI_Z \otimes \cP_n)(\eta) \| (\cI_Z \otimes \cQ_n)(\eta) \bigr) \leq  n D_\alpha^\infty(\cN \| \cM) + \frac{\log 1/\varepsilon^2}{\alpha -1} + \log \frac{1}{1-\varepsilon^2} \,.
    \end{equation}
    After taking the supremum over $\eta$, the statement follows by dividing both sides by $n$, choosing $\alpha = 1 + \frac{1}{\sqrt{n}}$, and sending $n \to \infty$. 
\end{proof}

The REAT from~\cref{cor_REAT} implies the EAT from~\cite{MetgerFawziSutterRenner2024}. To see this, let $\cG\in\CPTP(RE,ARE)$ satisfy $\tr_{AR}\circ\cG=\cT\circ\tr_R$, and define $\cS_A(\omega)=\id_A\otimes\tr_A\omega$. Then, choose $S_i=A^iE$, $K_i=A^{i-1}$, $\cA_i=\cI_{A^{i-1}}\otimes\cG$, $\cB_i=\cI_{A^{i-1}}\otimes(\cS_A\circ\cG)$, $\cN=\tr_R\circ\cG$, and $\cM=\cS_A\circ\cN$. 

\subsection{Generalized quantum Stein's lemma} \label{sec_GSL}
The generalized quantum Stein’s lemma, introduced by~\cite{brandao_Stein_10}, focuses on distinguishing many copies of a fixed quantum state from a set of alternative states, which may be correlated. Under suitable assumptions on these sets, it identifies the optimal type-II error exponent, for any fixed type-I error with the regularized relative entropy distance to the alternative set. A gap in the original proof was identified by~\cite{Berta2023gapinproofof} and independent complete proofs were given in~\cite{haya_stein_25,ludo25}. Since then, further generalizations have appeared~\cite{MSR_stein26,LBT26, girardi2026stein} and the original gap from~\cite{brandao_Stein_10} has been resolved by~\cite{MSR_stein26}.

The tools developed in our work also give an alternative proof of the generalized quantum Stein lemma with a more general null hypothesis. 
Let $\cS_n\subseteq \St(\cH^{\otimes n})$ be nonempty, closed and convex sets, such that $\cS_n \otimes \cS_m \subseteq \cS_{n+m}$ and $\cS_1$ contains a state $\tau>0$.\footnote{These are the same set axioms as used in~\cite{haya_stein_25}.} Furthermore, let us introduce the notation $D(\rho_n\|\cS_n) \coloneqq \min_{\sigma_n\in \cS_n} D(\rho_n\|\sigma_n)$ for $\rho_n \in \St(\cH^{\otimes n})$
and define the type-II error
\begin{equation}
    \beta_n^\varepsilon(\rho_n\|\cS_n) \coloneqq \inf_{0\leq T_n \leq \id
    }  \Big \{\sup_{\sigma_n \in \cS_n} \tr[\sigma_n T_n] : \tr[\rho_n T_n]\geq1-\varepsilon\Big \} \,.
\end{equation}
We consider sequences $(\rho_n)_{n\geq1}$ of states for which the limit
\begin{equation} \label{eq_assumption_limit_exists}
    R\coloneqq\lim_{n\to\infty}\frac1n D(\rho_n\|\cS_n)
\end{equation} exists.
We further assume the following strong converse: For every $\eta>0$ and every sequence of tests $0 \leq T_n \leq \id$,
\begin{equation} \label{eq_strong_converse_assumption}
    \sup_{\sigma_n\in\cS_n}\tr[\sigma_nT_n] \leq\ee^{-n(R+\eta)}
    \quad\Longrightarrow\quad \lim_{n \to \infty} \tr[\rho_nT_n]=0\, .
\end{equation}
These assumptions hold for iid~states (see~\cite{brandao_Stein_10}), but also permutation-invariant almost-iid states\footnote{Under the additional assumptions on the alternative set~\cite[Theorem 4.3]{MSR_stein26}}(see~\cite{MSR26,MSR_stein26}) or general stationary ergodic sequences (which can be shown from~\cite{Bjelakovic_2004}).

\begin{corollary}
Under the above assumptions, for every $\varepsilon \in(0,1),$
    \begin{equation}
        \lim_{n\to\infty} -\frac{1}{n} \log \beta_n^{\varepsilon}(\rho_n \|\cS_n) = \lim_{n\to \infty} \frac{1}{n} D(\rho_n \| \cS_n) \, .
    \end{equation}
\end{corollary}
\begin{proof}
    Define $c\coloneqq-\log\lambda_{\min}(\tau)$.
    The central step will be to construct a positive definite operator $Z_n=Z_n(\rho_n)$ satisfying 
    \begin{align}
        \sup_{\sigma_n\in \cS_n} \tr[\sigma_n Z_n] &\leq 1, \label{eq:Z1}\\
        \log Z_n &\leq cn\id,\label{eq:Z2}\\
        \tr[\rho_n \log Z_n]&\geq nR - o(n) \, . \label{eq:Z3}
    \end{align}
    This can be done by defining
    \begin{equation}\label{eq:an-wn}
        A_n \in \argmin_{A\in\{\frac{\sigma_n+\tau^{\otimes n}}{2}:\sigma_n\in\cS_n\}} D(\rho_n \|A) \qquad \textnormal{and} \qquad  W_n = \frac{1}{2} f_{0,A_n}(\rho_n)
    \end{equation}
    followed by
    \begin{equation} \label{eq:bn-zn}
        B_n = W_n + \ee^{-n} \id \qquad \textnormal{and} \qquad Z_n = \frac{B_n}{1+\ee^{-n}}\, ,
    \end{equation}
    where the function $f_{0,A}$ was introduced in~\cref{lem:derivative}. The minimum $A_n$ exists because the feasible set is compact and every feasible operator is bounded below by $\tau^{\otimes n}/2$.

    We start by showing~\cref{eq:Z1}. Since $A_n$ minimizes the objective, for any $\sigma_n \in \cS_n$, moving towards $(\sigma_n + \tau^{\otimes n})/2$ cannot decrease it. Therefore
    \begin{align}
        0 
        &\leq \frac{\di }{\di t} \, D\left(\rho_n\| A_n+t\left(\frac{\sigma_n+\tau^{\otimes n}}{2}-A_n \right)\right)\Biggr|_{t=0^+} \\
        \overset{\textnormal{\Cshref{lem:derivative}}}&{=} -\tr\left[f_{0,A_n}(\rho_n)\left(\frac{\sigma_n+\tau^{\otimes n}}{2}-A_n \right)\right] \\
        \overset{\textnormal{\Cshref{eq:an-wn}}}&{=} 1- \tr[W_n(\sigma_n+\tau^{\otimes n})] \, , \label{eq_derivative_LS}
    \end{align}
    where the final step uses $\tr[f_{0,A_n}(\rho_n)A_n]=1$. 
Since $W_n\geq 0 $ (see~\cref{lem:derivative}) we have
\begin{align} \label{eq_mid_step}
\sup_{\sigma_n \in \cS_n} \tr[\sigma_n W_n]\overset{\textnormal{\Cshref{eq_derivative_LS}}}{\leq} 1 \, .
\end{align}
Hence, 
    \begin{align}
\sup_{\sigma_n\in \cS_n} \tr[\sigma_n Z_n]
\overset{\textnormal{\Cshref{eq:bn-zn}}}{=} \frac{1}{1+\ee^{-n}} \sup_{\sigma_n\in \cS_n} \tr[\sigma_n (W_n +\ee^{-n}\id)]
\overset{\textnormal{\Cshref{eq_mid_step}}}&{\leq} 1 \, .
    \end{align}
 
To show~\cref{eq:Z2}, we use that by~\cref{lem:derivative} $f_{0,A_n}(\id) = A_n^{-1}$. Using $\rho_n\leq \id$ and $A_n \geq \tau^{\otimes n}/2$, we obtain 
\begin{align}
0
\leq W_n 
\overset{\textnormal{\Cshref{eq:an-wn}}}{=} \frac{1}{2} f_{0,A_n}(\rho_n)
\overset{\textnormal{\Cshref{eq:f_s-integral}}}{\leq}\frac{1}{2} f_{0,A_n}(\id)
= \frac{1}{2} A_n^{-1} 
\leq (\tau^{-1})^{\otimes n} 
\leq \ee^{cn} \id \, , \label{eq_step2_ts}
\end{align}
where the penultimate step uses that $t \mapsto t^{-1}$ is operator anti-monotone~\cite[Table~2.2]{Sutter_book}.
Since $c\geq 0 $, we find
\begin{align}
0 
\overset{\textnormal{\Cshref{eq:bn-zn,eq_step2_ts}}}{\leq} Z_n
\overset{\textnormal{\Cshref{eq:bn-zn}}}{=} \frac{W_n + \ee^{-n}\id}{1+\ee^{-n}}
\overset{\textnormal{\Cshref{eq_step2_ts}}}{\leq} \frac{\ee^{c n} + \ee^{-n}}{1+\ee^{-n}} \id
\leq \ee^{cn}\id \, .
\end{align}
Operator monotonicity of the logarithm then proves~\cref{eq:Z2}. 
    
Finally, to prove~\cref{eq:Z3}, note that~\cref{lem:derivative} gives us 
    \begin{equation}\label{eq:meas-log2}
        D_{\mathbb{M}}(\rho_n \|W_n) = \tr[\rho_n \log A_n] +\log 2 \, .
    \end{equation}
    To compare this with the ordinary relative entropy, note that the spectrum of $B_n$ lies in $[\ee^{-n},2\ee^{cn}]$. If we round each eigenvalue of $B_n$ upward to the next integer power of $\ee$, without changing its eigenspaces, the resulting operator $B_n'$ satisfies 
    \begin{equation}
        B_n \leq B_n' \leq \ee B_n \label{eq:order}
    \end{equation}
    and has at most $v_n = O(n)$ distinct eigenvalues. 
    Putting everything together then gives 
    \begin{align}
        \tr[\rho_n \log Z_n]
        \overset{\textnormal{\Cshref{eq:bn-zn}}}&{=} \tr[\rho_n \log B_n]-\log(1+\ee^{-n}) \\
        &= \tr[\rho_n\log\rho_n]
           -D(\rho_n\|B_n)-\log(1+\ee^{-n}) \\
        \overset{\textnormal{\Cshref{eq:order}}}&{\geq}  \tr[\rho_n\log\rho_n]
           -D(\rho_n\| B_n')
           -1-\log(1+\ee^{-n}) \\
        \overset{\textnormal{\Cshref{lem:pinch}}}&{\geq} \tr[\rho_n\log\rho_n]
           -D_{\mathbb{M}}(\rho_n\| B_n')
           -\log v_n-1-\log(1+\ee^{-n}) \\
        \overset{\textnormal{\Cshref{eq:bn-zn,eq:order}}}&{\geq}  \tr[\rho_n\log\rho_n]
           -D_{\mathbb{M}}(\rho_n\|W_n)
           -\log v_n-1-\log(1+\ee^{-n}) \\
        \overset{\textnormal{\Cshref{eq:meas-log2}}}&{=}  \tr[\rho_n\log\rho_n]
           -\tr[\rho_n\log A_n]
           -\log 2-\log v_n-1-\log(1+\ee^{-n}) \\
        &= D(\rho_n\|A_n)
           -\log v_n-\log 2-1-\log(1+\ee^{-n}) \\
        &\geq D(\rho_n\|\cS_n)-o(n) \\
        \overset{\textnormal{\Cshref{eq_assumption_limit_exists}}}&{\geq}   nR-o(n)\, ,
    \end{align}
    which proves~\cref{eq:Z3}.
    
    Define $Y_n = \frac{1}{n}\log Z_n$ and the test
    \begin{equation}
        T_n(x) \coloneqq \mathbf{1}_{[\ee^{nx},\infty)}(Z_n) = \mathbf{1}_{[x,\infty)}(Y_n)\, \quad \textnormal{for } x \in \R \, .
    \end{equation}
 Let $r<R$ and note that $T_n(r) \leq \ee^{-nr} Z_n$ implies 
\begin{align}
\sup_{\sigma_n \in \cS_n}\tr[\sigma_nT_n(r)]
\leq \ee^{-nr} \sup_{\sigma_n \in \cS_n}\tr[\sigma_n Z_n]
\overset{\textnormal{\Cshref{eq:Z1}}}{\leq} \ee^{-nr} \, .
\end{align}
Hence, it remains to show that $\tr[\rho_n (\id-T_n(r))]$ tends to zero. To show this, consider a second test $T_n(R+\eta)$ for $\eta>0$. Analogous to before,  
\begin{align}
\sup_{\sigma_n \in \cS_n}\tr[\sigma_nT_n(R+\eta)]
\leq \ee^{-n(R + \eta)} \sup_{\sigma_n \in \cS_n}\tr[\sigma_n Z_n]
\overset{\textnormal{\Cshref{eq:Z1}}}{\leq} \ee^{-n(R + \eta)} \, .
\end{align}
By the strong converse (see~\cref{eq_strong_converse_assumption}), we have 
\begin{align} \label{eq_SC_step}
\lim_{n \to \infty}\tr[\rho_n T_n(R+\eta)] = 0 \, .
\end{align}
On the other hand, \cref{eq:Z2} implies
    \begin{align}
        Y_n &\leq r \mathbf{1}_{(-\infty,r)}(Y_n) + (R+\eta) \mathbf{1}_{[r,R+\eta)}(Y_n) + c \mathbf{1}_{[R+\eta,\infty)}(Y_n) \\
        &= r (\id -T_n(r)) + (R+\eta) \big(T_n(r) - T_n(R+\eta)\big) + cT_n(R+\eta) \\
        &= (R+\eta)\id - (R+\eta-r)\big(\id-T_n(r)\big) + (c-R-\eta) T_n(R+\eta) \, . \label{eq_almost_stein_done}
    \end{align}
With this we find    
    \begin{align}
        R-o(1) 
        \overset{\textnormal{\Cshref{eq:Z3}}}&{\leq} \tr[\rho_n Y_n] \\
        \overset{\textnormal{\Cshref{eq_almost_stein_done}}}&{\leq} R +\eta - (R+\eta-r) \tr[\rho_n (\id-T_n(r))] + (c -R-\eta)\tr[\rho_n T_n(R+\eta)] \\
        \overset{\textnormal{\Cshref{eq_SC_step}}}&{\leq}R +\eta - (R+\eta-r) \tr[\rho_n (\id-T_n(r))] + o(1) \, ,
    \end{align}
    which rearranged yields
    \begin{equation}
        \limsup_{n \to \infty} \tr[\rho_n (\id-T_n(r))] \leq \frac{\eta }{R+\eta-r} \, .
    \end{equation}
    Letting $\eta \downarrow 0$ then proves achievability. Together with the strong converse given in~\cref{eq_strong_converse_assumption}, this proves the assertion.
\end{proof}

\paragraph{Independent work}
During the preparation of this work, we were aware of two alternative proofs of the continuity statement that were independently derived~\cite{gao2026,JW26}. 
The proofs~\cite{gao2026,JW26} use substantially different techniques than our proof presented in this manuscript.

\paragraph{Acknowledgements}
We acknowledge GPT 6 Astra for assistance in developing the proof.
D.S.~thanks Omar Fawzi, Tony Metger, and Renato Renner for discussions on the continuity of the regularized channel R\'enyi divergence several years ago.
L.S.~acknowledges support from the Quantum Center at ETH Zurich and SNSF Grant No. 200021E\_232425.
\appendix

\section{Justification of~\cref{eq:simple_direction}} \label{app_justification_simple_limit}
Using the monotonicity of the R\'enyi divergence in $\alpha$~\cite[Theorem~7]{MLDSFT13} gives
\begin{align}
\lim_{\alpha \uparrow 1} D_{\alpha}^{\reg}(\cE \| \cF)  \leq D^{\reg}(\cE \| \cF) \, . \label{eq_step0}
\end{align}
For any fixed $r \in \N$ we find by definition of the channel relative entropy
\begin{align}
\frac{1}{r} D\big( \cE^{\otimes r} \| \cF^{\otimes r} \big)
&=\frac{1}{r} \max_{\nu \in \St((R A)^{\otimes r})} D\big( \cE^{\otimes r}_R(\nu) \| \cF^{\otimes r}_R(\nu) \big) \nonumber \\
&= \frac{1}{r} D\big( \cE^{\otimes r}_R(\nu^\star) \| \cF^{\otimes r}_R(\nu^\star) \big)  \nonumber \\
&= \lim_{\alpha \uparrow 1} \frac{1}{r} D_{\alpha}\big( \cE_R^{\otimes r}(\nu^\star) \| \cF_R^{\otimes r}(\nu^\star) \big) \, , \label{eq_step1}
\end{align}
where $\nu^\star \in \St((RA)^{\otimes r})$ denotes the maximizer. The final step uses the fact that R\'enyi relative entropy is continuous as ${\alpha \uparrow 1}$~\cite[Theorem~5]{MLDSFT13}.
Using the fact that the R\'enyi relative entropy is additive under tensor products gives
\begin{align}
 \lim_{\alpha \uparrow 1} \frac{1}{r} D_{\alpha}\big( \cE_R^{\otimes r}(\nu^\star) \| \cF_R^{\otimes r}(\nu^\star) \big)
 &=  \lim_{\alpha \uparrow 1} \lim_{k\to \infty} \frac{1}{rk} D_{\alpha}\Big(  \big(\cE_R^{\otimes r}(\nu^\star)\big)^{\otimes k} \| \big( \cF_R^{\otimes r}(\nu^\star) \big)^{\otimes k} \Big) \nonumber \\
 &\leq \lim_{\alpha \uparrow 1} \lim_{n\to \infty} \frac{1}{n} \max_{\tau \in \St((R A)^{\otimes n})} D_{\alpha}\big(\cE_R^{\otimes n}(\tau) \|  \cF_R^{\otimes n}(\tau)  \big) \nonumber \\
 &=\lim_{\alpha \uparrow 1} D_{\alpha}^{\reg}(\cE \| \cF) \label{eq_step2} \, ,
\end{align}
with $n=rk$, where the penultimate step uses that $(\nu^\star)^{\otimes k} \in \St((RA)^{\otimes n})$. Combining~\cref{eq_step1,eq_step2} gives
\begin{align}
\frac{1}{r} D\big( \cE^{\otimes r} \| \cF^{\otimes r} \big)
\leq \lim_{\alpha \uparrow 1} D_{\alpha}^{\reg}(\cE \| \cF) \, .
\end{align}
Since this is true for any $r\in \N$ we find $D^{\reg}(\cE \| \cF) \leq \lim_{\alpha \uparrow 1} D_{\alpha}^{\reg}(\cE \| \cF)$ by considering the limit $r \to \infty$. Combining this with~\cref{eq_step0} proves the assertion.

\bibliographystyle{arxiv_no_month}
\bibliography{bibliofile}

\end{document}